\documentclass[11pt]{article}
\usepackage{geometry}
\usepackage{graphicx}
\usepackage{amsmath}
\usepackage{amssymb}
\usepackage[all]{xy}

\newtheorem{theorem}{Theorem}[section]
\newtheorem{corollary}[theorem]{Corollary}
\newtheorem{lemma}[theorem]{Lemma}
\newtheorem{claim}[theorem]{Claim}
\newtheorem{definition}[theorem]{Definition}

\newtheorem{observation}[theorem]{Observation}
\newtheorem{proposition}[theorem]{Proposition}

\def\squarebox#1{\hbox to #1{\hfill\vbox to #1{\vfill}}}
\newcommand{\qed}{\hspace*{\fill}\vbox{\hrule\hbox{\vrule\squarebox{.667em}\vrule}\hrule}\smallskip}
\newenvironment{proof}{\noindent{\bf Proof:~~}}{\(\qed\)}

\newcommand{\xhdr}[1]{\noindent \textbf{#1.} \noindent}
\newcommand{\comment}[1]{}

\newcommand{\eps}{\varepsilon} 

\newcommand{\full}[1]{#1}
\newcommand{\short}[1]{}

\begin{document}
\title{Incentives and Coordination in Bottleneck Models \footnote{We thank Refael Hassin, Moshe Haviv, Ella Segev and the participants of the ``Queuing and Games'' Seminar at TAU for useful comments on this manuscript.}}

\author{
Moshe Babaioff
\thanks{
Microsoft Research.
}
\and{
Sigal Oren
\thanks{
Ben-Gurion University of the Negev.
}}
}

\maketitle

\begin{abstract}
We study a variant of Vickrey's classic bottleneck model. In our model there are $n$ agents and each agent strategically chooses
when to join a first-come-first-served observable queue. Agents dislike standing in line and they take actions in discrete time steps: we assume that each agent has a cost of $1$ for every time step he waits before joining the queue and a cost of $w>1$ for every time step he waits in the queue. At each time step a single agent can be processed. Before each time step, every agent observes the queue and strategically decides whether or not to join, with the goal of minimizing his expected cost. 
 
In this paper we focus on symmetric strategies which are arguably more natural as they require less coordination. This brings up the following twist to the usual price of anarchy question: what is the main source for the inefficiency of symmetric equilibria? is it the players' strategic behavior or the lack of coordination?

We present results for two different parameter regimes that are qualitatively very different:
(i) when $w$ is fixed and $n$ grows, we prove a tight bound of $2$ and show that the entire loss is due to the players' selfish behavior (ii) when $n$ is fixed and $w$ grows, we prove a tight bound of $\Theta \left(\sqrt{\frac{w}{n}}\right)$ and show that it is mainly due to lack of coordination: the same order of magnitude of loss is suffered by any symmetric profile. 

\end{abstract}

\section{Introduction}
William Vickrey is well known for his fundamental contributions to Mechanism Design, including the celebrated second price auction. However, his contributions were not limited to Mechanism Design. In a seminal paper from 1969, Vickrey \cite{vickrey1969congestion} identifies bottlenecks as a significant reason for traffic congestion. Bottlenecks are short road segments with a fixed capacity. Once the capacity is reached a queue is formed. Vickrey presents a rush hour model -- there are many employees that need to get to work around the same time, all need to cross the same bridge. It is assumed that they have some cost associated with each minute they arrive early to work and a potentially different cost associated with each minute they arrive late to work. Moreover, they have a different (and usually higher cost) for every minute they wait in traffic. Given these costs the employees need to decide when to leave for work in order to minimize their total cost. 
Similar timing decisions appear in many other situations: e.g., deciding when to go to the doctor or when to enter a traffic intersection. 

Vickrey's paper has inspired a line of work analyzing variants of this model both in economics and transportation theory (for example, \cite{arnott1990economics,arnott1993structural,ben1991dynamic}). Vickrey, as common in the literature, assumes that the population is continuous. This assumption considerably simplifies the analysis and is motivated by the observation that in large populations the externalities that any single agent is imposing on the rest are negligible. The clear downside of this assumption is that it fails to model scenarios with relatively small population. Moreover, models of discrete and continuous populations can behave differently, as the analysis of the price of anarchy (PoA) in routing games with high-degree polynomials demonstrated. Specifically, while for continuous population the PoA is linear in the degree, it is exponential for discrete population (\cite{roughgarden2002bad,Roughgarden:2005atomic, AwerbuchAE13}). 

{Unlike Vickrey, we
study a discrete population model. One of the choices we need to make is whether the agents observe the state of the traffic or not (this makes no difference in Vickrey's model where the population is modeled as a continuum). The few works that did study discrete variants of Vickrey's model (for example \cite{levinson2005micro,otsubo2008vickrey}), all made the assumption that no traffic information is provided (i.e., the queue is unobservable). However, technological advancements such as webcams that are installed over bridges, as well as mobile apps that provide information regarding traffic, call for focusing the analysis around the case that commuters \emph{do} have some aggregate information on the state of traffic, e.g, 
they observe the length of the queue. 
Hence, in our model we assume that the agents do observe the traffic's state (i.e., observable queue).

\xhdr{Our Model} 
We study a stylized variant of Vickrey's model to allow us to focus on two issues that were not explored in the original model: a discrete population and an observable queue.\footnote{Our minor simplifications of Vickrey's model include an assumption that the agents can only join the queue after some starting time.}

Formally, we have $n$ agents that, starting at time $0$, need to get a service which is offered by a first-come-first-served queue.  
Time progresses in discrete steps and in each step, each agent needs to decide whether to enter the queue at that time step or stay outside. When multiple agents decide to enter the queue simultaneously, they are ordered by a uniform random permutation. 
In our model, agents observe everything, and in particular, they observe the length of the queue and the set of agents that are still outside the queue.

For each agent, starting at time $0$, the cost per time step for waiting before joining the queue is normalized to $1$, and the cost per 
time step for waiting in the queue is $w>1$ (agents dislike waiting in the queue). At each time step a single agent can be processed.

Our model is inspired by traffic related scenarios similar to Vickrey's rush hour scenario, such as the following one: the first game of the 2018 NBA finals at the Oracle Arena has just ended, and the audience wants to get home to San Francisco. To get home they should all cross the Bay Bridge that has a limited capacity (a bottleneck). While each person wants to get home as soon as possible, he dislikes standing in traffic. So, he should strategically decide when to leave the stadium and try to cross the bridge, aiming to minimize his discomfort. Hanging out around the stadium is an option that is costly, but not as much as standing in traffic. Fortunately for the audience, there are cameras installed over the bridge and apps that constantly broadcast the traffic state, and they can observe it and make their decision accordingly. 

\xhdr{Solution Concept} As our game is symmetric and all agents are ex-ante the same, we focus on equilibria in symmetric randomized strategies that are anonymous. In such profiles all agents play the same randomized strategy that does not depend on the identities of the other agents. 
Observe that as agents are symmetric, equilibrium in asymmetric strategies requires different agents to play differently although they are ex-ante symmetric. This requires the agents to coordinate on which strategy each of them will play. Thus, symmetric strategies are arguably more natural than other, more general, strategies\footnote{In  \cite{otsubo2008vickrey}, Otsubo and Rapoport are making a similar argument in favor of symmetric strategies.}.

Moreover, our focus will be on {\em stationary strategies} that do not depend on the time step, but rather only on the state -- the number of agents that are outside the queue 
as well as  
the number of agents that are in the queue.  
More formally, we consider symmetric Nash equilibria in anonymous stationary strategies, or {\em symmetric strategies} for short. Under such strategies, 
for any state there is some defined probability of
entering the queue, and that probability is used by all agents.
In Appendix \ref{app-eq-existance} we prove that such equilibria always exist. 

The discrete population and discrete time assumptions
make the analysis of our model quite challenging. In particular computing equilibrium strategies for a game with $n$ agents requires solving $n$ polynomial equations of degrees increasing from $1$ to $n$. Otsubo and Rapoport \cite{otsubo2008vickrey} are among the few that studied discrete population variants of Vickrey's model. 
They have only provided a complicated algorithm to numerically compute symmetric mixed Nash equilibrium rather than obtaining closed-form expressions.

Other approaches that were taken in similar models include analyzing the fluid limit of the system \cite{jain2011concert}, and computing an equilibrium for the continuous time model by solving differential equations \cite{juneja2013concert}. We take a different approach and instead of explicitly computing a symmetric equilibrium, present asymptotically tight upper and lower bounds on the {social} cost of any symmetric equilibrium.\footnote{Doing so alleviates the need to precisely compute symmetric equilibria and the need to determine if the game has a \emph{unique} symmetric equilibrium or not.}

\xhdr{Our Results} 
We study the efficiency of symmetric equilibria in terms of the social cost, which is simply the sum of the players' costs. We present bounds that hold for any symmetric equilibria, any $n$ and any $w$.\footnote{While we prove bounds that hold for any $n$ and any $w$ (see Theorems \ref{thm:opt-fixed-w} and  \ref{thm:opt-large-w}) our focus in the presentation is on the asymptotic social cost of symmetric equilibria, when either $w$ or $n$ gets large.} 
Thus, we establish tight asymptotic price of stability and price of anarchy results for symmetric equilibria. Moreover,
our lower bounds imply price of anarchy lower bounds for general Nash equilibria. \footnote{Similarly to the ``Fully Mixed Nash Equilibrium Conjecture'' \cite{gairing2003extreme} we suspect that in our game symmetric equilibria are in fact the worst equilibria.}

We first analyze the ratio between the social cost of any symmetric equilibrium and the social cost of the optimal solution. We observe that whenever $1<w \leq 2$ (the cost of waiting one time step in line is at most twice the cost of waiting outside), the unique symmetric equilibrium is for all agents to enter the queue immediately, and the total social cost is $w\cdot n\cdot (n-1)/2 $. 
On the other hand, when agents enter sequentially, the social cost is only $n\cdot (n-1)/2$ (this is also the minimal social cost when  we do not impose any restrictions on the strategies, in particular the strategies can be non-anonymous and non-stationary.)

Thus, in this case the ratio between the social costs of the unique symmetric equilibrium and the optimal solution is $w \leq 2$. Loosely speaking a similar bound of $2$ also holds when we take a fixed $w$ which is considerably smaller than $n$, but the proof is much more involved. A bit more formally, we 
prove the following result:

\begin{theorem}
\label{thm:intro-POA-fixed-w}
Fix $w>2$. As $n$ approaches infinity the ratio between the social costs of any symmetric equilibrium and the optimal solution is approaching $2$.
\end{theorem} 

Usually in scenarios such as commuting to work or deciding when to head to the bridge after a game, the number of agents is relatively large while the 
normalized cost of waiting a unit of time in traffic, $w$, is relatively small. Theorem \ref{thm:intro-POA-fixed-w} tells us that in such cases the loss of efficiency in a symmetric Nash equilibrium is constant and relatively low, only 2. 

We next consider the other extreme parameter regime, where the cost $w$ is large relative to $n$. Such scenarios might arise in cases where either people do not care so much about getting the service early (for example, taking a routine medical check-up or running some non-urgent bureaucratic errand) or they have an arbitrary high cost for arriving simultaneously with others to receive the service. For example, one can think of a traffic intersection as providing a service for which simultaneous entry might cause an accident and has a very high cost. For the regime that $n$ is small relative to $w$ we obtain qualitatively different results than those for the regime that $n$ is large relative to $w$. This provides an additional confirmation for the value of studying models of discrete population. 

\begin{theorem}
\label{thm:intro-POA-fixed-n}
Fix $n$. As $w$ approaches infinity the ratio between the social costs of any symmetric equilibrium and the {(unrestricted)} optimal solution 
is approaching $2\cdot \sqrt{\frac{w}{n}}$, up to an additive term of $O (\frac{1}{\sqrt{n}})$. 
\end{theorem}
The theorem shows that when $w$ is significantly greater than $n$, the multiplicative efficiency loss is very high (grows asymptotically as $\sqrt{w}$). Essentially, the issue with symmetric equilibria is that when the cost of standing in line is very high, the players are so horrified at the prospect of waiting in the queue
that they enter the queue at a very low probability. Thus, the service is actually idle most of the time. To reduce this kind of inefficiency, society came up with symmetry breaking mechanisms such as traffic lights or doctors appointments (see \cite{cayirli2003outpatient}), which provide a much needed coordination.
Such a coordination mechanism induces a sequential order of entry, no player ever waits in the queue and the social cost is minimal. Moreover, for large $w$ ($w>2n$), this is actually an (asymmetric) equilibrium.  

Theorem \ref{thm:intro-POA-fixed-n} shows that symmetric equilibria are highly inefficient, but what is the main source of the inefficiency of symmetric equilibria: is it the players' strategic behavior or is it the lack of coordination imposed by symmetric strategies? As far as we know, no prior work has tried to separate between the loss of efficiency of symmetric equilibria due to strategic behavior and due to the symmetry requirement.
To answer this question we bound the cost of the symmetric optimal solution and provide bounds on the ratio between the cost of any symmetric equilibrium and the symmetric optimal solution. We derive two different asymptotic bounds, depending on whether $n$ or $w$ are fixed.

\begin{theorem}
\label{thm:intro-symmetric-POA-fixed-w}
Fix $w>2$. As $n$ approaches infinity the 
ratio between the social costs of any symmetric equilibrium and the \emph{symmetric} optimal solution is approaching $2$.
\end{theorem} 
In fact, as the ratio between the social costs of any symmetric equilibrium and the optimal solution is also approaching $2$, this implies that as $n$ grows large the social cost of the symmetric optimal solution is approaching the social cost of the (unrestricted) optimal solution, and both have essentially the same gap from any symmetric equilibrium. Thus, for this case, we conclude that the \emph{main source of inefficiency of symmetric Nash equilibria is the strategic selfish behavior of the players.}

When considering the regime in which $n$ is fixed but large, while $w$ grows to infinity, a different picture emerges. Intuitively, since the cost of having two or more agents entering the queue simultaneously is so large, to avoid this cost, any profile of symmetric strategies must use entry probabilities that are low enough to ensure that the expected number of agents that join the queue at each step is much lower than $1$. As a result, the agents will wait for a long time before anyone enters the queue, which implies a high social cost.
 
This creates a large gap between the social costs of the symmetric optimal solution and the (unrestricted) optimal solution. Interestingly, the gap is of the same magnitude as the gap between the worst symmetric equilibrium and the optimal solution. We show:
\begin{theorem}
\label{thm:intro-symmetric-POA-fixed-n}
Fix $n$. As $w$ approaches infinity the ratio between the social costs of any symmetric equilibrium and the \emph{symmetric} optimal solution is approaching
$\frac{3}{2\sqrt{2}}  \approx 1.06$, up to an additive term of  $O (\frac{1}{\sqrt{n}})$.
\end{theorem} 
Recall that Theorem \ref{thm:intro-POA-fixed-n}  is showing that as $w$ approaches infinity the ratio between the social costs of any symmetric equilibrium and the optimal solution is approaching $2\cdot \sqrt{\frac{w}{n}}$. When we combine the two theorems we conclude that
in the case that $w$ goes to infinity, the \emph{main source of inefficiency of the symmetric Nash equilibrium is the lack of coordination in symmetric randomized strategies}. Nevertheless, there is still some small constant loss that does not vanish and is due to incentives, yet it dwarfs compared to the loss of $2\cdot \sqrt{\frac{w}{n}}$ (which tends to infinity as $w$ grows to infinity) that is due to lack of coordination.

\xhdr{Related Work}
Bottleneck models were studied in both the traffic science literature and the economics literature. Arnott et al. \cite{arnott1990economics} provide economic analysis of Vickrey's bottleneck model and also consider how tolls can reduce the cost associated with strategic behavior in such models. Later papers extended the model to handle more general pricing schemes (for example \cite{arnott1993structural, daganzo2000pareto}). In \cite{arnott1991does} Arnott et al. consider giving traffic information to commuters in a continuous population model in which commuters need to choose when to leave and which route to take. The new twist is that travel time  is affected by unexpected events such as accidents or bad weather that the commuters can get information about. Arnott et al. reach the conclusion that providing the commuters perfect information about these unpredictable events can eliminate the inefficiency resulting from them. Other variants of the model that were studied more recently include: heterogeneous commuters \cite{van2011congestion} and the effects of congested bottlenecks on the roads leading to them \cite{lago2007spillovers}.

The literature on strategic queuing is also related to the current paper. Similar to Vickrey's model, this line of work has also originated from a paper first published in 1969. In his seminal paper, Naor was the first to introduce both economic and strategic considerations into the queuing literature \cite{naor1969regulation}. Up till then queuing theory mainly focused on the efficiency of queues. The most well known model in classic queuing theory is the M/M/1 queue model, where there is one server and the jobs arrive according to a Poisson distribution and have an exponentially distributed service time. According to Naor's model, the service has a price and the ``jobs'' need to decide whether to join the queue or not. This gave rise to a new area called strategic queuing which studies the users' behavior in different queuing systems under various assumptions (see \cite{hassin2003queue,hassin2016queue} for extensive surveys).

In general, the literature on strategic queuing has traditionally focused on models of unobservable queues as these are easier to analyze (see chapter 2 of \cite{hassin2016queue} for a survey of recent works on observable queues). Hassin and Roet-Green \cite{Hassin-Roet-Green} bridge the gap between observable queues and non-observable queues by presenting and analyzing a natural model in which the agents have the option to pay to see the length of the queue. Most of the works in strategic queuing (both on observable and unobservable queues) consider games in which each player arrives at some time and needs to immediately decide whether to join the queue or not. In this setting, the paper of Kerner~\cite{Kerner11} applies a solution concept similar to ours in studying symmetric equilibrium joining probabilities for an {M/G/1} observable queue.

The more elaborate model in which the player's strategy is to choose an arrival time with the goal of minimizing his waiting time was first suggested in \cite{glazer1983m} for unobservable queues. \cite{glazer1983m} studied a model in which agents can choose to join the queue before its opening time (early arrivals) while later \cite{hassin2010equilibrium} showed that the efficiency of the equilibrium can be sometimes increased by  disallowing early arrivals.
 Discrete time and discrete population versions of this model were later studied in \cite{lariviere2004strategically, rapoport2004equilibrium} that concentrated on symmetric mixed Nash equilibria for this unobservable queue model. 
A recent work \cite{Lingenbrink-unobservable} studies a setting in which while the queue is unobservable the service provider can observe the queue and give the agents some information regarding the queue.

A specific line of papers in strategic queuing which is similar both in intuition and in formalism to our model is on the so called ``concert queuing game''. This game was first defined in  \cite{jain2011concert} and was later studied in follow up papers such as \cite{juneja2013concert, honnappa2015strategic}. In this game concert attendants wish to get home as soon as possible once the concert ends but they dislike standing in traffic.
The main distinctions between the concert queuing game and our model are in the assumptions on the observability of the queue and its processing time (in our model, processing time is fixed, while in the other model it is distributed according to some distribution).
In each of 
\cite{jain2011concert} and \cite{juneja2013concert} the authors use different analysis techniques to establish a bound of $2$ on the price of anarchy for their model.  

The strategic queuing literature includes a few papers dealing with price of anarchy. The first one was \cite{haviv2007price} which studied the price of anarchy of a multi server system where the players strategically choose which queue to join (without observing the queues' states first). In     \cite{gilboa2014price} Gilboa-Freedman et al. provide a price of anarchy bound for Naor's model \cite{naor1969regulation} where the queue is observable but the strategy of each player is limited to the one-time decision whether to join the queue or not. 

The paper of Fiat et al. \cite{fiat2007efficient} also considers strategic entry by selfish players -- players that need to broadcast on a joint channel. The model in that paper is fundamentally  different than ours, as simultaneous entry results in all players failing to enter the channel, rather than a formation of a queue as in our model.
 	
\xhdr{Paper Outline} We start by formally presenting our model and some useful observations. Next, in Section \ref{sec:2player} we study the two player case as a warm-up. 
We then present our upper and lower bounds on the cost of any symmetric equilibria in Section \ref{sec:symmetric-NE}, 
Theorems \ref{thm:intro-POA-fixed-w} and \ref{thm:intro-POA-fixed-n} follow from these results. 
Finally, we present bounds for the social cost of symmetric optimal strategies in Section \ref{sec:symmetric-OPT}, Theorems \ref{thm:intro-symmetric-POA-fixed-w} and \ref{thm:intro-symmetric-POA-fixed-n} follow from these bounds.

 \section{Model and Preliminaries}  \label{sec:model}
 There is a set $N$ of $n$ identical agents and time is discrete. \full{  We use $t=1,2,3\ldots$ to denote a time step.

At each time step $t$ an agent has two possible actions; enter the queue (1) or wait (0). 
Agents that did not enter yet are said to be {\em outside} the queue, and agents that entered but are still in the queue are {\em in line}. An agent that enters at time $t$ will be processed after every agent that entered at time $t'<t$. If multiple agents decide to enter at time $t$, they will enter the queue in a uniform random order.  More formally, at}\short{At} every time step $t$ the following sequence takes place:

\begin{itemize}
\item Each agent decides 
whether or not to enter the queue (possibly using randomization).
\item After the agents make their decisions, all agents that have decided to enter the queue are added to the end of the queue in a random order.
\item If the queue is not empty then the first agent in the queue is processed. The rest of the agents in the queue incur a cost of $w>1$.
\item Every agent outside the queue incurs a cost of $1$.
\end{itemize}

The goal of each agent is to minimize his expected cost.
We use $G(n;w)$ to denote a game with $n\geq 2$ agents and waiting cost per unit of $w>1$. 

In general, a strategy of an agent needs to specify the probability of entry at each {\em history}, such a history specifies the time, the realized action of each agent at every prior time, and the randomization results whenever multiple agents that enter at the same time are ordered at the end of the queue. 
Our focus in this paper is on  anonymous stationary strategies -- strategies that do not depend on the time step or the identity of agents, such strategies will only depend on a summary statistics specified by the (anonymized) {\em state} of our game.    
A {\em state} of our game is defined as a pair $(m,k)$, where there are $m\geq 1$ agents that are still outside the queue and $k\geq 0$ agents in the queue, and $m+k\leq n$. 
We formally define anonymous stationary strategies as follows:  

\begin{definition}[anonymous stationary strategies in $G(n;w)$]
A strategy of an agent is an {\em anonymous stationary strategy} if for any history it specifies a probability that an agent enters the queue that is only a function of the state $(m,k)$. We denote that probability of entry at a state $(m,k)$ by $q_{m,k} \in [0,1]$ and the probability at state $(n,0)$ by $q_n$.
\end{definition}

We use $S=(S_1,S_2,\ldots, S_n)$ to denote a profile of strategies. A profile $S$ of anonymous stationary strategies is {\em symmetric} if all players use the same strategy ($S_i=S_j$ for every $i,j\in N$). 
We are interested in Nash equilibria of the game, in such equilibria each player minimizes his cost, given the strategies of the others. \full{To be brief we sometimes refer to a ``symmetric equilibrium in anonymous stationary strategies'' simply as {\em symmetric equilibrium}.} In Appendix \ref{app:general-strategies} we discuss general strategies and show that if a profile is an equilibrium with respect to 
anonymous stationary strategies, it is also an equilibrium with respect to any arbitrary strategies that might depend on the time or on the identities of the agents. For this reason, for the rest of the paper we will only consider deviations to anonymous stationary strategies.

We denote the expected cost of every player in a symmetric equilibrium $S$ in the game $G(n;w)$ by $c_{n,w}(S)$. When $w$ and $S$ are clear from the context we simplify the notation to $c_n$. We denote the social cost of strategy profile $S$ by $C_{n,w}(S) = n \cdot c_{n,w}(S)$. As we will see, it is useful to extend this notation to sub-games as well. We denote the sub-game that starts with a state $(m,k)$ by $G(m,k;w)$. For a symmetric equilibrium $S$ we denote the cost of each of the $m$ players that are outside the queue by $c_w(m,k;S)$. When $S$ and $w$ are clear from the context we will use the shorter notation $c(m,k)$. With this notation we have that $ c_{n,w}(S)=c(n,0)$.

\subsection{Basic observations}
If at state $(m,k)$ a player enters with a non-trivial probability $q_{m,k}\in (0,1)$ then he must be indifferent between entering and waiting at that step. 
When not indifferent, the player will choose to either enter the queue or wait with probability one. Our analysis of the social cost of symmetric Nash Equilibria heavily depends on this observation. Thus, it is useful to first work out the expressions for a cost of a player $i$ that joins the queue with probability $1$ at state $(m,k)$ and the cost of player $i$ that joins the queue with probability $0$ at state $(m,k)$. We define the two costs as $c^1(m,k;q)$ and $c^0(m,k;q)$ respectively, where we assume that in state $(m,k)$ all the players but $i$ enter with probability $q$ and at any other state all the players play according to some symmetric strategy profile $S$. We now give expressions for the two costs for every symmetric profile $S$: 

\begin{observation} \label{obs:eq-costs-1}
For every $w>1$, $q\in[0,1]$, $m \geq 1$ and $k\geq 0$: $c^1(m,k;q) = \frac{m-1}{2} \cdot  q \cdot w + k\cdot w$.
\end{observation}
\begin{proof}
The cost has two parts: a cost associated with players that joined the queue together with player $i$ but were randomly assigned to be before him in line, and a cost associated with the $k$ players already standing in line. The second cost is simply $k\cdot w$. The first part of the cost can be computed as follows: for each player $j \neq i$ we can define a random binary variable $X_j$ which equals $1$ if $j$ enters the queue before $i$ and $0$ otherwise. We observe that $Pr[X_j] = \frac{q}{2}$, since with probability $q$ it will enter the line and since the player in the queue are order randomly with probability $1/2$ it will be ordered before player $i$. Now, by linearity of expectation we get that the expected number of players that will be before player $i$ in line is $\frac{m-1}{2} \cdot  q$.
\end{proof}

\begin{observation}\label{obs:eq-costs-0}
For every $q\in[0,1]$, $m\geq 1$ and $k\geq 1$:  
\begin{align*}
c^0(m,0;q) &= \frac{1}{1-(1-q)^{m-1}} + \frac{1}{1-(1-q)^{m-1}} \cdot \sum_{i=1}^{m-1} {{m-1}\choose{i}} q^i \cdot (1-q)^{m-1-i} \cdot c(m-i,i-1;S) \\
c^0(m,k;q) &= 1+\sum_{i=0}^{m-1} {{m-1}\choose{i}} q^i \cdot (1-q)^{m-1-i} \cdot c(m-i,k+i-1;S)
\end{align*}
\end{observation}
\begin{proof}
It holds that 
\begin{align*}
c^0(m,0;q) &= 1+ (1-q)^{m-1} \cdot c^0(m,0;q)+\sum_{i=1}^{m-1} {{m-1}\choose{i}} q^i \cdot (1-q)^{m-1-i} \cdot c(m-i,i-1;S) \\
\implies c^0(m,0;q)
&= \frac{1}{1-(1-q)^{m-1}} + \frac{1}{1-(1-q)^{m-1}} \cdot \sum_{i=1}^{m-1} {{m-1}\choose{i}} q^i \cdot (1-q)^{m-1-i} \cdot c(m-i,i-1;S)
\end{align*}
This is so as the agent pays $1$ for waiting, then with probability $ (1-q)^{m-1}$ no other agent enters, and we are back in the same situation, and if exactly $i>1$ other agents enter, an event that happens with probability ${{m-1}\choose{i}} q^i \cdot (1-q)^{m-1-i} $, the agent will pay $c(m-i,i-1;S)$ since one of the agents that entered will be processed. 

Similar computation holds for $k\geq 1$, with the exception that if no player enters, one of the agents in the non-empty queue will be processed, hence we get the second equation. 
\end{proof}

Next, we claim that a symmetric equilibrium always exists. This is not a priori clear as our strategies cannot be easily defined as a mixed of pure anonymous stationary strategies.

\begin{theorem}
Fix any $n\geq 2$ and $w>1$.
There exists a symmetric equilibrium in anonymous stationary strategies in the game $G(n;w)$. 
\end{theorem}

The proof idea is rather simple, we show that a symmetric equilibrium exists in every sub-game $G(m,k;w)$, by induction on $m+k$. Essentially, in every such sub-game $G(m,k;w)$ we show that it is either the case that there exists $q^*$ such that $c^0(m,k;q^*) = c^1(m,k;q^*)$ and hence there exists a random symmetric equilibrium in which agents enter with probability $q^*$ at this state, or for every value of $q$, $c^0(m,k;q) < c^1(m,k;q)$ or $c^1(m,k;q) < c^0(m,k;q)$ in which case there is symmetric equilibrium in which all players either do not enter the queue or all players do enter the queue, at this state. The formal proof is slightly more complicated since the function $c^0(m,0;q)$ is not continuous at $q=0$. The complete proof can be found in Appendix \ref{app-eq-existance}.

We observe that for small values of $w$ ($w\leq 2$) the unique symmetric equilibrium outcome is for all players to enter immediately.\footnote{Note that the uniqueness is for every sub-game that is actually played - as all players enter immediately, states $(m,k)\neq (n,0)$
are never reached and the play of the game as well as the payoffs are independent of how players intend to play in such sub-games $G(m,k;w)$.} 

\begin{observation}
 	\label{obs:sym-ne-cost-small-w}
 	For any $n\geq 2$. If $w\in [1,2]$ then in the game $G(n;w)$ the unique symmetric equilibrium with anonymous stationary strategies is for all agents to enter with probability $q_n=1$ and the social 
 	cost is $w \cdot \frac{n(n-1)}{2}$.
 \end{observation}
 \begin{proof}
 	We will show a stronger claim: in the game $G(n;w)$ for $w\in [1,2]$ a player's dominant strategy is $q_n=1$. To this end, it is enough to show that for any number of players $0\leq k \leq n-1 $ that joined the queue in the first step (after the randomness was realized) the player prefers to join the queue with probability $1$ in the first step. To see why this is the case, note that if $k$ players joined then the player's cost for joining the queue is 
 	$\frac{k}{2} \cdot w <k$. On the other hand, if he chooses not to join the queue his cost will be at least $k$ as he will need to wait for at least $k$ time steps before leaving the queue.
 \end{proof}
 
  We focus our equilibrium analysis on the case that $w>2$. We now show that when $w>2$ it is no longer the case that a player in a symmetric equilibrium prefers to join the queue with probability $1$. Furthermore, in case the queue is empty, deterministically not entering the queue is also not a best response to the strategies of others in a symmetric equilibrium.
  
  \begin{observation}
  	\label{obs:enterance-prooob-large-w}
  	For any $n\geq 2$, $w>2$ and for any $(m,k)$ such that $m+k\leq n$, $m\geq 2$ and $k\geq 0$, in any symmetric equilibrium $0\leq q_{m,k}<1$ and $0< q_{m,0}<1$.
  \end{observation}
  \begin{proof}
  	We show that for $(m,k)$ such that $m\geq 2$ and $k\geq 0$, in no symmetric equilibrium it holds that $q_{m,k}=1$.
  	Indeed, the cost of an agent by entering when all other agents enter with probability $1$ is  
  	$\frac{m-1}{2} \cdot w + k\cdot w> m+k-1$, while if the agent waits till all others are processed and then enter, 
  	his cost is only $m+k-1$. Thus, in any symmetric equilibrium $q_{m,k}<1$.  The claim that $0< q_{m,0}$ is immediate since if all other players do not enter, a player wants to enter immediately. 
  \end{proof}
 
The following observation characterizes the cost of symmetric Nash equilibria based on the observations above:
\begin{observation} \label{obs:NE-costs}
	Fix $w>2$ and $n\geq 2$. Given a symmetric equilibrium in the game $G(n;w)$, the cost of some player $i$ which is outside the queue in the state $(m,k)$ satisfies the following:
	\begin{itemize}
	\item For $m\leq n$ it holds that $c^1(m,0;q_{m}) = c^0(m,0;q_{m})$.
	\item For $k\geq 1$, if $q_{m,k}\in (0,1)$ then $c(m,k) = c^0(m,k;q_{m,k}) =  c^1(m,k;q_{m,k})$. Otherwise, $q_{m,k} = 0$ and $c(m,k) = c^0(m,k;0)$.
	\end{itemize}
\end{observation}

 \section{Warm-up - The 2 Players Case}
 \label{sec:2player}
 To get some intuition, before turning into the general case with many agents it 
is instructive to consider first the simple case of only $2$ players. For this case we obtain an exact expression for the players' costs in the unique symmetric equilibrium.

Note that a single player will always join the queue ($q_{1,0}=1$) and for this reason $c(1,0) = 0$. Now, to compute a symmetric equilibrium we only need to 
compute the probability that the agents enter when both are still outside ($q_{2,0}$).
 
\begin{claim} \label{clm-2players-ne-prob}
 For $n=2$ players, if $w>2$, there exists a unique symmetric equilibrium $S$ in anonymous stationary strategies in 
 which each player plays the strategy: $q_{2,0} = \sqrt{2/w}$, 
 $q_{1,0}=1$. The social cost for both players in this equilibrium is $C_{2,w}(S) = \sqrt{2w}$.  
 \end{claim}
 \begin{proof}
 To compute $q_{2,0}$ we use Observation \ref{obs:NE-costs} implying that a player in the game $G(2,0;w)$ is indifferent between joining the queue with probability $1$ and staying outside. By Observation \ref{obs:eq-costs-1} 
 the cost of the first option is $c^1(2,0;q) = w\cdot  q/2$ and by Observation \ref{obs:eq-costs-0} the latter option has a cost of $c^0(2,0;q) = \frac{1}{q} + \frac{1}{q} \cdot q \cdot c(1,0) = \frac{1}{q}$, note that this is the expected number of steps till the other player enters, when this player waits outside. Putting this together we get that $w\cdot  \frac{q}{2} = \frac{1}{q}$, and thus $q_{2,0} = \sqrt{2/w}$. The social cost to both players is $2\cdot (w\cdot q_{2,0}/2) = w\cdot \sqrt{2/w} = \sqrt{2w}$.  
\end{proof}

We compare the cost at Nash equilibrium against the cost of the optimal solution. For two players, the cost of the optimal solution is simply $1$ as one of the players will enter first and pay a cost of $0$ and the other will enter second and pay a cost of $1$. This implies the following corollary:
\begin{corollary}
In the game $G(2;w)$ the ratio between the social costs of the unique\footnote{Note that when $w>2$ the two players game admits exactly three equilibria: the two optimal equilibria in which one player enters after the other, and the symmetric random equilibrium we discussed. Thus, our result is both a price of anarchy result for unrestricted equilibria and a price of stability result for symmetric equilibria.} symmetric Nash equilibrium and optimal solution is $\sqrt{2w}$.
\end{corollary}

This relatively large gap that grows with $w$ leads us to ask what is the source of this gap -- is it due to strategic behavior, or to the lack of coordination imposed by symmetric strategies? To answer this question we compute the minimal cost when all agents are required to use the same strategy and use anonymous stationary strategies, which are not necessarily an equilibrium. We note that even just for $2$ players computing the optimal symmetric solution is simple yet not completely trivial, as it requires computing the minimum of a function which is the ratio of two polynomials. As the number of players increases this becomes more complicated and hence instead of directly computing the optimal symmetric strategy we will compute bounds on its cost.

 As in the Nash equilibrium, once an agent is the only one outside, he clearly enters immediately. Thus, we only need to compute the probability of each agent entering, assuming both agents are outside (denoted $p_2$).
  
    \begin{claim} \label{clm-2players-opt-prob}
    	For $n=2$ players, if $w> 1$, the symmetric anonymous stationary strategy that minimizes the social cost is: $p_2= \frac{\sqrt{2w-1}-1}{w-1}$ and $p_1=1$. 
    	The social cost for this profile is
    	$OPT(2,w) = \frac{w+1}{\sqrt{2w-1}+1}$.
    \end{claim}
  
\begin{proof}
  	Let $p=p_2$. It is easy to see that: 
  	\begin{align*}
  		OPT(2,w) &= (1-p)^2 (2+OPT(2,w)) + 2p(1-p) \cdot 1 + p^2 \cdot w \\
  		&= \frac{(1-p)^2 \cdot 2 + 2p(1-p) \cdot 1 + p^2 \cdot w}{1-(1-p)^2}=  \frac{2-2p+ p^2 \cdot w}{p(2-p)} =
  		  	\frac{1}{p} + \frac{w\cdot p -1}{2-p}
  	\end{align*}
  	
	Clearly, $p=0$ is not optimal. 
	To find $p\in (0,1)$ that minimizes this we take the derivative by $p$ and check when it equals zero:
	\begin{equation*}
	\frac{2p^2(w-1)+4p-4}{p^2(p-2)^2} = 0 
	\end{equation*}
	The unique solution for this equation in $(0,1)$  is $p=\frac{\sqrt{2w-1}-1}{w-1}$ and it is easy to verify that this is indeed the minimum in $(0,1)$. Moreover, the left side derivative at $p=1$ is positive, so $p=1$ is not a minimizer. We conclude that $p=\frac{\sqrt{2w-1}-1}{w-1}$  is the unique minimizer of this function in $[0,1]$, for any $w>1$. 
	For this $p$ the value of $OPT(2,w)$ is 
	\begin{align*}
	OPT(2,w) = \frac{1}{p} + \frac{w\cdot p -1}{2-p}  = \frac{w-1}{\sqrt{2w-1}-1} + \frac{w\cdot \frac{\sqrt{2w-1}-1}{w-1} -1}{2-\frac{\sqrt{2w-1}-1}{w-1}}  
	\end{align*}
	After simplifying this expression a bit we get that:
	$OPT(2,w) = \frac{w+1}{\sqrt{2w-1}+1}.$
  \end{proof}

The following corollary is easily derived from the above claim:
\begin{corollary}
 In the game $G(2;w)$ the ratio between the social costs of the unique symmetric Nash equilibrium and optimal solution in symmetric strategies is approaching $2$ as $w$ approaches infinity.
 \end{corollary}
 
We conclude that for large values of $w$, there is a huge loss for insisting on symmetric profiles: the optimal cost grows from $1$ in an asymmetric optimum, to about $\sqrt{\frac{w}{2}}$ in the symmetric optimum. An additional, much smaller, loss of factor $2$ comes from further requiring the symmetric profile to be an equilibrium. Our goal in this paper is to understand the source of inefficiency of symmetric Nash equilibria for any $n$. We present a separation between the cost ratio of symmetric Nash equilibria and the symmetric optimal solution when $n$ is fixed and $w$ is large and the case that $w$ is fixed but $n$ is large.


\section{Bounds on the Cost of Symmetric Nash Equilibria}
\label{sec:symmetric-NE}
In this section we provide bounds on the cost of symmetric Nash equilibria in any profile of anonymous stationary strategies, these bounds hold for any $n$ and any $w$.
We present 
two types of bounds, each will be tight for a different regime of the parameters $n$ and $w$, and we use these bounds to prove Theorem \ref{thm:intro-POA-fixed-w} and Theorem \ref{thm:intro-POA-fixed-n}.
We first present a bound that is useful when $w$ is relatively small compared to $n$.

\begin{theorem}\label{thm:opt-fixed-w}
	For every $w>2$, $n\geq 2 $ and symmetric equilibrium $S$:
	$$n-1\leq c_{n,w}(S) \leq n + w \cdot O(\ln (n))$$ 
\end{theorem}
Clearly the above bound is asymptotically tight whenever $w=o\large(\frac{n}{\ln(n)}\large)$. This implies that when $w=o(\frac{n}{\ln(n)})$, for a sufficiently large value of $n$, the social cost of any symmetric equilibrium is about $n^2$. 
Denote by $SC(n,w)$ the social cost of the optimal solution. Recall that in the optimal solution the players enter sequentially and hence the cost of the optimal solution is $SC(n,w)=n(n-1)/2$. The ratio between the cost of any Nash equilibrium and the cost of an optimal solution is essentially $2$, proving Theorem \ref{thm:intro-POA-fixed-w}. Formally: 
\begin{corollary} \label{cor:nash-2}
	For every fixed $w>2$ and every $\eps>0$ there exists $n_0^w(\eps)$ such that for any $n>n_0^w(\eps)$ for every symmetric equilibrium $S$ it holds that 
	$$2\leq C_{n,w}(S)/SC(n,w)\leq 2+\eps$$
\end{corollary}

Next, we give a different bound which will be tight for the case of large enough $n$, and $w$ that goes to infinity. We show that 
the social cost of any symmetric equilibrium is essentially approaching $n\cdot \sqrt{w\cdot n}$. Formally:  

\begin{theorem}\label{thm:opt-large-w}	
	For every $w>2$, $n\geq 2$ and any symmetric equilibrium $S$:
	$$ (1-\eps_n(w))^{n-1} \cdot \sqrt{w\cdot n- w\cdot 2\ln n} \leq c_{n,w}(S) \leq  \frac{e}{e-1}\cdot n + (1+\eps_n(w)) \cdot \sqrt{w \cdot n + 2w\sqrt{n-1}}$$ 
	for some decreasing function $\eps_n(w) \leq 1$ that for any fixed $n$, converges to $0$ as $w$ grows to infinity.
\end{theorem}

In this case if $n = o(\sqrt w)$ we get that the social cost of any symmetric equilibrium is about $n \cdot \sqrt{w\cdot n}$ and hence the ratio between the costs of any symmetric equilibrium and the optimal solution (which has cost of $n(n-1)/2$) is approaching $2\cdot \sqrt{\frac{w}{n}}$, proving Theorem \ref{thm:intro-POA-fixed-n}: 
\begin{corollary}
	Fix any $\eps>0$. 
	There exists $n_0(\eps)$ such that for any $n>n_0(\eps)$
	there exist $w^{n}(\eps)$ such that for any $w>w^{n}(\eps)$ it holds that for any symmetric equilibrium $S$:   
	$$ 
	(2-\eps) \sqrt{\frac{w}{n}}  \leq \frac{C_{n,w}(S)}{SC(n,w)} \leq (2+\eps) \sqrt{\frac{w}{n}}
	$$	
\end{corollary}

\short{We prove Theorems \ref{thm:opt-fixed-w} and \ref{thm:opt-large-w} in Appendix \ref{sec:fixed-w} and Appendix \ref{sec:large-w}, respectively.}
To prove \short{the two theorems} \full{Theorems \ref{thm:opt-fixed-w} and \ref{thm:opt-large-w}} we separately prove an upper bound and a lower bound on the cost of any symmetric equilibrium. We derive the lower bound for Theorem \ref{thm:opt-fixed-w} by basic observations on our game. 
Our proof for the lower bound for Theorem \ref{thm:opt-large-w} is much more involved, and it utilizes the upper bound of Theorem \ref{thm:opt-large-w}.  
\short{In the next section we present a sketch of the main argument common to the proofs of the two upper bounds.
\subsection{A Sketch of the Upper Bounds Proofs}
} 
\full{In the next section we present a general recipe for proving both of our upper bounds. Then, we provide in Section \ref{sec:fixed-w} a complete proof of Theorem \ref{thm:opt-fixed-w}. The proof of Theorem \ref{thm:opt-large-w} can be found in  Section \ref{sec:large-w}.
\subsection{A General Recipe for Proving Upper Bounds}
}

The proofs of the two upper bounds of Theorem \ref{thm:opt-fixed-w} and Theorem \ref{thm:opt-large-w} follow a similar structure based on an induction. 
We define a general notion of a ''nice upper bound function'' and show that this notion is useful in bounding the social cost in any symmetric equilibrium without explicitly computing it.   
	Each of our two upper bounds is proven by a using  different upper bound function $\phi(m,k)$, each exhibits some nice properties captured by the following definition.
\begin{definition} \label{def-nice-upper}
A candidate upper bound function $\phi(m,k)$ is \emph{nice} if:
\begin{enumerate}
\item For every $m\geq 1$ and $k\geq 0$: $\phi(m,k) - \phi(m,0) \geq k$.
\item $\phi(m,k)$ is monotone in the following sense: $\phi(m,k) \geq \phi(m',k')$ for every $m\geq m' $ and $m+k\geq k'+m'$, that is, moving agents from outside the queue to the queue decreases the upper bound on the cost.
\end{enumerate}
\end{definition}

Next, we show  
how the definition of nice candidate upper bound function can assist us in proving upper bounds for our game.

\begin{proposition} \label{prop:nash-nice-ub} 
Consider a candidate upper bound function $\phi(m,k)$ that is nice, and a game $G(n;w)$ such that $n\geq 2$. Assume that for every symmetric equilibrium it holds that for every $m \geq 1$, $k\geq 0$ such that $m+k \leq n-1$, we have that $c(m,k) \leq \phi(m,k)$. Then:
\begin{itemize}
\item For every $m+k = n$ such that $m\geq 1$ and  $k\geq 1$, it holds that $c(m,k) \leq \phi(m,k)$.
\item $c_n \leq \frac{1}{1-(1-q_n)^{n-1}} + \phi(n -1,0)$. 
\end{itemize}

\end{proposition}

\begin{proof}
For the first statement, when $m+k=n$ such that $m,k\geq 1$, we know that an agent can always decide to wait for $k$ steps till the $k$ agents in the queue are being served. Thus, we have that 
\begin{align*}
c(m,k) \leq k + \max_{0\leq i \leq m} c(m-i,i) \leq k + \max_{0\leq i \leq m} \phi(m-i,i)
\end{align*}
where the last transition is by using the proposition assumptions. Now, the monotonicity condition on $\phi(m,k)$ (the second condition of Definition \ref{def-nice-upper}) comes in handy as it tells us that $\max_{0\leq i \leq m} \phi(m-i,i) = \phi(m,0)$. The proof is completed by using the first requirement of a nice candidate upper bound function which tells us that $k + \phi(m,0) \leq \phi(m,k)$.

For the second statement, it is easy to see that by Observations \ref{obs:NE-costs} and \ref{obs:eq-costs-0},
	as well as the fact that for any $q\in (0,1)$ it holds that ${1-(1-q)^{n-1}} = \sum_{i=1}^{n-1} {{n-1}\choose{i}} q^i \cdot (1-q)^{n-1-i}$, that:	
\begin{align*}
c(n,0) \leq \frac{1}{1-(1-q_{n,0})^{n-1}} + \max_{1\leq i \leq n} c(n-i,i-1) 
\end{align*}
The proposition now follows from using 
 the fact that $c(m,k) \leq \phi(m,k)$ for every $m \geq 1$ and $k\geq 0$ such that $m+k \leq n-1$ and the monotonicity of $\phi(m,k)$ to 
observe that $ \max_{1\leq i \leq n} c(n-i,i-1)\leq \phi(n -1,0)$.
\end{proof}

Using Proposition \ref{prop:nash-nice-ub} as part of an induction on the state space $(m,k)$ we derive the two upper bounds that hold for any $n\geq 2$ and $w\geq 2$}:
\begin{itemize}
\item Using $\phi(m,k) =  m+ k+ \sqrt{\frac{w}{2}}  + \sum_{i=1}^{m-1} \frac{w}{2\cdot i}$ we prove (\short{Appendix \ref{sec:fixed-w-upper-bound}}\full{Section \ref{sec:fixed-w-upper-bound} below}) the upper bound claimed in Theorem \ref{thm:opt-fixed-w}, that bound is tight for a fixed $w$ and a large $n$.
This proof relies on a corollary of our lower bound $c(m,0) \geq m-1$ stating that for every $m\geq 1$, $q_{m,0} > \frac{2}{w}$.  
\item Using $\phi(m,k) = \frac{e}{e-1}  \cdot (m+k) + (1+\eps_n(w)) \cdot \left( \sqrt{w} \cdot \sqrt{m+2\sqrt{m-1}}\right)$ we prove (\short{Appendix} \full{Section} \ref{sec:large-w-upper-bound}) the upper bound claimed in Theorem \ref{thm:opt-large-w}, that bound is tight for a fixed $n$ and a large $w$. Here, we distinguish between the cases that $q_{m,0} \geq \frac{1}{m-1}$ and $q_{m,0} < \frac{1}{m-1}$ and handle each one separately.
\end{itemize}

\full{
In the next two sections we provide the proofs of the bounds in Theorem \ref{thm:opt-fixed-w} and Theorem \ref{thm:opt-large-w}. In Section \ref{sec:fixed-w-upper-bound} we prove the bounds for Theorem \ref{thm:opt-fixed-w} when $w$ is fixed, and in Section \ref{sec:large-w-lower-bound} we prove the bounds for Theorem \ref{thm:opt-large-w} when $n$ is fixed.
}

\full{\subsection{Bounds for Small $w$ (for Theorem \ref{thm:opt-fixed-w}) \label{sec:fixed-w}}

We first present a simple lower bound on the cost of any symmetric equilibrium. This lower bound will turn out to be tight in case $w$ is fixed and $n$ is large. Next, we use this lower bound together with Proposition \ref{prop:nash-nice-ub} to get a tight upper bound for small values of $w$.

\subsubsection{A Simple Lower Bound} \label{sec:fixed-w-lower-bound}
We next prove that the cost of any symmetric Nash equilibrium is at least $n(n-1)$. 
This shows that the social 
cost of any symmetric equilibrium is at least twice the cost of the asymmetric optimum (in which the agents enter sequentially).

\begin{claim}
	\label{clm:simple-lb}
	For any $w> 2$ and any natural numbers $m\geq 1,k\geq 0$ it holds that in every symmetric equilibrium $c(m,k)\geq m+k-1$. 
	In particular, $c_n=c(n,0)\geq n-1$.
\end{claim}
\begin{proof}
	We prove the claim by induction on n=$m+k$. 
	If $m+k=1$ the claim trivially holds. Assume that the claim holds for any $m',k'$ such that $m'+k'=n-1$, we prove the claim for any $m,k$ such that $m+k=n$. 	
	
	We first handle the case that $k\geq 1$. 
	By Observation \ref{obs:NE-costs} we have that $c(m,k)$ equals the cost of the strategy in which the agent first waits one round: $c(m,k) = c^0(m,k;q_{m,k})$. Thus by Observation \ref{obs:eq-costs-0} and the induction hypothesis:
	\begin{align*}
	c(m,k) & = 1+\sum_{i=0}^{m-1} {{m-1}\choose{i}} q_{m,k}^i \cdot (1-q_{m,k})^{m-1-i} c(m-i,k-1+i) \\ 
	&\geq 1+\sum_{i=0}^{m-1} {{m-1}\choose{i}} q_{m,k}^i \cdot (1-q_{m,k})^{m-1-i} (m+k-2)\\ 
	&= 1+ (m+k-2) = m+k-1
	\end{align*}
	
	We next prove the claim for the case that $k=0$. For this case by applying Observations \ref{obs:NE-costs} and \ref{obs:eq-costs-0} together with the induction hypothesis, we get that:
	
	\begin{align*}
	c_n &= \frac{1}{1-(1-q_n)^{n-1}} + \frac{1}{1-(1-q_n)^{n-1}} \cdot \sum_{i=1}^{n-1} {{n-1}\choose{i}} q_n^i \cdot (1-q_n)^{n-1-i} c(n-i,i-1) \\ &  \geq 1 + \frac{1}{1-(1-q_n)^{n-1}} \cdot (n-2) \sum_{i=1}^{n-1} {{n-1}\choose{i}} q_n^i \cdot (1-q_n)^{n-1-i}  
	\\ &  = 1 + \frac{1}{1-(1-q_n)^{n-1}} \cdot (n-2) (1-(1-q_n)^{n-1}) = n-1 
	\end{align*}
	which proves the claim. 	
\end{proof}

An immediate Corollary from the claim about is that the social cost of any symmetric equilibrium is at least $n(n-1)$: 
\begin{corollary}
	For any $w> 2$ and any $n\geq 2$ in every symmetric equilibrium the social 
	cost is at least $n(n-1)$, which is twice the cost of the asymmetric optimum. 
\end{corollary}

From the lower bound on the cost of any symmetric equilibrium we can derive a lower bound on the players' entrance probabilities. 

\begin{corollary} \label{cor-p-is-large}
	For any $w> 2$ and any $n\geq 2$ it holds that  in every symmetric equilibrium $q_n\geq \frac{2}{w}$. Also, for every $m\geq 1$ and $k\geq 0$ it holds that
			$q_{m,k}\geq \frac{2}{w} \cdot \left(1- \frac{k(w-1)}{m-1} \right) $.
\end{corollary}

\begin{proof}
	Note that the cost of entering with probability $1$ is always greater than or equal to the cost the player exhibits, thus it holds that:
	$q_{m,k}\cdot w\cdot \frac{m-1}{2}+k\cdot w \geq c(m,k)$. Now, by Claim \ref{clm:simple-lb} it holds that 
	$c(m,k)\geq m+k-1$ we get that 
	\begin{align*}
	q_{m,k}\cdot w\cdot \frac{m-1}{2}+k\cdot w \geq c(m,k) \geq m+k-1 
	\end{align*}
	or equivalently 
	\begin{align*}
	q_{m,k}\geq \frac{2}{w} \cdot \frac{m+k-1 - k\cdot w}{m-1}   = \frac{2}{w} \cdot \left(1- \frac{k(w-1)}{m-1} \right)
	\end{align*}
	To complete the proof observe that the lower bound on $q_{m,k}$ is greater than $0$ 
	when $1- \frac{k(w-1)}{m-1} >0$. 
	In particular, it is positive when $k=0$ and thus $q_n \geq  \frac{2}{w}$. 
\end{proof}

\subsubsection{An Upper Bound for Small $w$} \label{sec:fixed-w-upper-bound}
In this section we will use Proposition \ref{prop:nash-nice-ub} with $\phi(m,k) =  m+ k+ \sqrt{\frac{w}{2}}  + \sum_{i=1}^{m-1} \frac{w}{2\cdot i}$ to show that $c_n=c(n,0) \leq n +w \left(\frac{2+\ln n}{2} \right)$.

\begin{proposition}\label{prop:nash-UB-linear-w}
For every $w> 2$ and for every $n\geq 2$ in every symmetric equilibrium:
\begin{align*} 
c_n=c(n,0) \leq n +w \left(\frac{2+\ln n}{2} \right)
\end{align*}  
\end{proposition}
\begin{proof}
We prove by induction over $n$ that for every $m\geq 1$ and $k\geq 0$ such that $m+k \leq n$, it holds that $c(m,k) \leq \phi(m,k)$. For the base $n=2$ it is easy to see that $c(2,0) =\sqrt{\frac{w}{2}}\leq \phi(2,0)$ and $c(1,1)=1 \leq \phi(1,1)$. For the induction step, we assume the induction holds for every $m+k \leq n-1$ such that $m\geq 1$ and $k\geq 0$ and prove it holds for every $m+k=n$ such that $m\geq 1$ and $k\geq 0$. 

It is easy to see that $\phi(m,k) =  m+ k+ \sqrt{\frac{w}{2}}  + \sum_{i=1}^{m-1} \frac{w}{2\cdot i}$ is a nice candidate upper bound function. Hence, we can use Proposition \ref{prop:nash-nice-ub} to get that:
\begin{itemize}
\item For every $m+k = n$ such that $m\geq 1$ and  $k\geq 1$, it holds that $c(m,k) \leq \phi(m,k)$.
\item $c_n \leq \frac{1}{1-(1-q_n)^{n-1}} + \phi(n -1,0)$. 
\end{itemize}
Claim \ref{clm:fixed-w-expected-time} below shows that for every $n\geq 2$, $w > 2$ and every symmetric equilibrium, the next inequality
	holds.
\begin{align*}
\frac{1}{1-(1-q_{n,0})^{n-1}} + \phi(n -1,0) \leq \phi(n,0)
\end{align*}

Combining these two claims we get that $c_n\leq \phi(n,0)$ as needed. Thus, we get that for every $w> 2$ and for every $n\geq 2$ for every symmetric equilibrium $S$:
\begin{align*}
c_n \leq n+ \sqrt{\frac{w}{2}}  + \sum_{i=1}^{n-1} \frac{w}{2\cdot i}
\end{align*}

To complete the proof we note that $\sqrt{\frac{w}{2}}\leq \frac{w}{2}$ for $w>2$, and using standard arguments on the sum of a harmonic series we get that:
\begin{align*}
c_n \leq n + \sqrt{\frac{w}{2}} + \sum_{i=1}^{n-1} \frac{w}{2\cdot i} 
\leq n + \frac{w}{2} + \frac{w}{2}\cdot \left(\ln(n-1) +1\right) 
\leq n +w \left(\frac{2+\ln n}{2} \right)
\end{align*}

\end{proof}

We now prove that $\frac{1}{1-(1-q_{m,0})^{m-1}} + \phi(m -1,0) \leq \phi(m,0)$:
\begin{claim} \label{clm:fixed-w-expected-time}
For $\phi(m,k) =  m+ k+ \sqrt{\frac{w}{2}}  + \sum_{i=1}^{m-1} \frac{w}{2\cdot i}$, every $n\geq 2$, $w > 2$ and every symmetric equilibrium:
\begin{align*}
\frac{1}{1-(1-q_{m,0})^{m-1}} + \phi(m -1,0) \leq \phi(m,0)
\end{align*}
\end{claim}
\begin{proof}
We use Corollary \ref{cor-p-is-large} stating that $q_{m,0} > \frac{2}{w}$.
Now, by using a simple auxiliary lemma from the appendix (Lemma \ref{lem-big-prob}) we get that:
\begin{align*}
\frac{1}{1-(1-q_{m,0})^{m-1}} < \frac{1}{1-(1-\frac{2}{w})^{m-1}
\leq} \frac{e^{\frac{2(m-1)}{w}}}{e^{\frac{2(m-1)}{w}}-1} = 1+ \frac{1}{e^{\frac{2(m-1)}{w}}-1}\leq 1+ \frac{w}{2(i-1)}
\end{align*}
Where the rightmost inequality follows from the fact that  by Taylor expansion $e^x > 1+x$ and hence
$\frac{1}{e^{\frac{2(i-1)}{w}}-1} < \frac{1}{\frac{2(i-1)}{w}} = \frac{w}{2(i-1)}$. It is easy to see that:
\begin{align*}
\left(1+ \frac{w}{2(m-1)} \right)+ \left(m-1 + \sqrt{\frac{w}{2}}  + \sum_{i=1}^{m-2} \frac{w}{2\cdot i} \right)= m+ \sqrt{\frac{w}{2}}  + \sum_{i=1}^{m-1} \frac{w}{2\cdot i}
\end{align*}
and the claim follows.
\end{proof}

From Proposition \ref{prop:nash-UB-linear-w} we derive the following corollary which clearly holds when $\frac{n}{\ln n}>>w$.
\begin{corollary}\label{cor:ub-fixed-w}
 	For every fixed $w>2$ and every  $\eps>0$ there exists $n_0^w(\eps)$ such that for every $n > n_0^w(\eps)$ we have that for every symmetric equilibrium $S$ it holds that $c_n(S) \leq (1+\eps)n$ and 
 	thus the social cost  
 	satisfies  $C_{n,w}(S)\leq (1+\eps)n^2$. 
 \end{corollary}}
\full{\subsection{Bounds for Large $w$ (for Theorem \ref{thm:opt-large-w})} \label{sec:large-w}
In this section we prove the lower bound and upper bound required for Theorem \ref{thm:opt-large-w}. Recall that Theorem \ref{thm:opt-large-w} states that for every $w>2$, $n\geq 2$ and any symmetric equilibrium $S$:
$$ (1-\eps_n(w))^{n-1} \cdot \sqrt{w\cdot n- w\cdot 2\ln n} \leq c_{n,w}(S) \leq  \frac{e}{e-1}\cdot n + (1+\eps_n(w)) \cdot \sqrt{w \cdot n + 2w\sqrt{n-1}}$$ 
for some decreasing function $\eps_n(w) \leq 1$ that for any fixed $n$, converges to $0$ as $w$ grows to infinity. In Section \ref{sec:large-w-upper-bound} we prove the upper bound and in Section \ref{sec:large-w-lower-bound} we prove the lower bound.

\subsubsection{Upper Bounds for Large $w$}  \label{sec:large-w-upper-bound}
The upper bound on the cost of symmetric equilibria that we prove in Proposition \ref{prop:nash-UB-linear-w} grows linearly in $w$ and logarithmically in $n$.
In this section we prove two upper bounds that grow much slower as a function of $w$, only as $\sqrt{w}$, in the expense of a larger growth in $n$.
Both our bounds hold for any $n\geq 2$. The first one holds for any $w>2$ and gets a coefficient of $2$ on the $\sqrt{w}$ term, 
while the second bound get a coefficient that converges to $1$ as $w$ grows, this rate of growth in $\sqrt{w}$ is asymptotically tight, as we prove a matching lower bound (see Corollary \ref{cor:lb-large-w}).

\begin{proposition} \label{prop:nash-upperbound-large}
	For every  $n\geq 2$ :
	\begin{itemize}
	\item For every $w>2$ and every symmetric equilibrium $S$: $c(n,0) \leq \frac{e}{e-1}  \cdot n + 2\cdot \left( \sqrt{w} \cdot \sqrt{n+2\sqrt{n-1}}\right)$. 
	\item For every $0< \eps < 1$, there exists $w_0^n(\eps)$ such that for any $w \geq w_0^n(\eps)$, and every symmetric equilibrium $S$:
	$$c(n,0) \leq \frac{e}{e-1}  \cdot n + (1+\eps) \cdot \left( \sqrt{w} \cdot \sqrt{n+2\sqrt{n-1}}\right)$$
	\end{itemize}
\end{proposition}
Recall that $c_{n,w}(S)=c(n,0)$, thus, The proposition implies the required upper bound for Theorem \ref{thm:opt-large-w}.

\begin{proof}
 We prove both bounds by induction using Proposition \ref{prop:nash-nice-ub}. We use the upper bound function $\phi_\eps(m,k) = \frac{e}{e-1}  \cdot (m+k) + (1+\eps) \cdot \left( \sqrt{w} \cdot \sqrt{m+2\sqrt{m-1}}\right)$ for the second statement and $\phi_1(m,k)$ for the first statement. It is easy to see that $\phi_\eps(m,k)$ is a nice candidate upper bound function according to Definition \ref{def-nice-upper} and thus we will be able to use Proposition \ref{prop:nash-nice-ub} as part of our induction.

In particular, we prove by induction over $n$ that 
\begin{itemize}
\item For every $w>2$ and every symmetric equilibrium $S$, for every $m\geq 1$ and $k \geq 0$ such that $m+k \leq n$, it holds that $c(m,k) \leq \phi_1(m,k)$.
\item For every $0< \eps < 1$, there exists $w_0^n(\eps)$ such that for any $w \geq w_0^n(\eps)$, and every symmetric equilibrium $S$: for every $m\geq 1$ and $k \geq 0$ such that $m+k \leq n$, it holds that $c(m,k) \leq \phi_\eps(m,k)$.
\end{itemize}

It is easy to see that the base case $n=2$ holds for both statements. This is because for any $w \geq 2$ and any $0 \leq \eps\leq 1$ we have that $c(2,0) =\sqrt{\frac{w}{2}}\leq \phi_\eps(2,0)$, $c(1,1)=1 \leq \phi_\eps(1,1)$. For the induction step, we assume the induction holds for every $m+k \leq n-1$ such that $m\geq 1$ and $k\geq 0$ and prove it holds for every $m+k=n$ such that $m\geq 1$ and $k\geq 0$. 

As $\phi_\eps(m,k) = \frac{e}{e-1}  \cdot (m+k) + (1+\eps) \cdot \left( \sqrt{w} \cdot \sqrt{m+2\sqrt{m-1}}\right)$ is a nice candidate upper bound function, we can use Proposition \ref{prop:nash-nice-ub} to get that:
For every $0< \eps \leq 1$, there exists $w_0^n(\eps)$ (in particular $w_0^n(1) =2$ ) such that for any $w \geq w_0^n(\eps)$ 
\begin{itemize}
\item For every $m\geq 1$, $k\geq 1$, such that $m+k = n$, it holds that $c(m,k) \leq \phi_\eps(m,k)$.
\item $c_n \leq \frac{1}{1-(1-q_n)^{n-1}} + \phi_\eps(n -1,0)$. 
\end{itemize}
Thus, to complete the induction's proof we use the statement that $c_n \leq \frac{1}{1-(1-q_n)^{n-1}} + \phi_\eps(n -1,0)$ to show that $c_n < \phi_\eps(n,0)$.

We will distinguish between two cases. First, if $q_n \geq \frac{1}{n-1}$,
then $$\frac{1}{1-(1-q_n)^{n-1}} \leq \frac{e^{(n-1) \cdot 1/(n-1)}}{e^{(n-1) \cdot 1/(n-1)} -1} = \frac{e}{e-1}$$ (the simple proof is provided in Lemma \ref{lem-big-prob} in the appendix). Thus,
\begin{align*}
c_n &\leq \frac{e}{e-1} \cdot+ \frac{e}{e-1} \cdot(n-1)+ (1+\eps) \cdot( \sqrt{w(n-2)}+\sqrt{w}) \\
&\leq \frac{e}{e-1} \cdot n + (1+\eps) \cdot( \sqrt{w(n-1)}+\sqrt{w})
\end{align*} and the claim holds.

The more complicated case is when $q_n < \frac{1}{n-1}$ which we handle next. For this case, by auxiliary Lemma \ref{lem-small-probability} that can be found in the appendix, we have that 
\begin{align*}
\frac{1}{1-(1-q_n)^{n-1}} \leq \frac{2}{2-q_n \cdot (n-2)} \cdot \frac{1}{q_n\cdot (n-1)} 
\end{align*}  
By the assumption that $q_n < \frac{1}{n-1}$ we have that 
$\frac{1}{1-(1-q_n)^{n-1}} \leq 2 \cdot \frac{1}{q_n\cdot (n-1)}$ which, as we will see later, suffices to prove the first statement of Proposition \ref{prop:nash-upperbound-large} without any further conditions on $w$ except for $w\geq 2$. For the second statement we rely on the following proposition that we prove in Appendix \ref{app-prob-bounds}:  

	\begin{proposition} \label{prp:nash:zero}
		Fix $n\geq 2$. For any $0< \eps \leq  1$ there exists $w_0^n(\eps)$ such that for any $w\geq w_0^n(\eps)$ in every symmetric equilibrium:
		\begin{itemize}
			\item For any $ 2 \leq m \leq n$ it holds that 
			$q_{m,0} \cdot (m-1) \leq \eps$. 
			\item For any $k\geq 1$ and $m\geq 1 $ such that $k+m \leq n$ it holds that $q_{m,k} = 0$.
		\end{itemize}
	\end{proposition}
	
	Thus, we have that there exists $w_0^n(\eps)$ such that for any $w>w_0^n(\eps) \geq w_0^{n-1}(\eps)$ and any 
	$2 \leq m\leq n$ we have that $q_{m,0} \cdot (m-1) \leq \eps$.
	This implies that $q_n \cdot (n-2) \leq \eps$ and thus:
	\begin{align*}
	\frac{1}{1-(1-q_n)^{n-1}} \leq \frac{2}{2-\eps} \cdot \frac{1}{q_n\cdot (n-1)} = \left(1+\frac{\eps}{2-\eps}\right) \cdot \frac{1}{q_n\cdot (n-1)}
	\end{align*}
	Finally, by using the assumption that $\eps \leq 1$ we get that 
	\begin{align*}
	\frac{1}{1-(1-q_n)^{n-1}} \leq \frac{1+\eps}{q_n\cdot (n-1)}
	\end{align*}
	
	Putting this together with our upper bound on $c_n$ we get that:
	\begin{align*}
	c_n \leq \frac{1+\eps}{(n-1)\cdot q_n}+ \frac{e}{e-1} \cdot(n-1)+ (1+\eps) \cdot( \sqrt{w(n-2)}+\sqrt{w})
	\end{align*}
	Recall that $c_n = c^1(n,0;q_n) = \frac{n-1}{2} \cdot q_n \cdot w$, hence $q_n (n-1)= \frac{2c_n}{w}$. By plugging this in the previous equation we get that: 
	\begin{align*}
	&c_n \leq \frac{(1+\varepsilon)w}{2c_n}+ \frac{e}{e-1} \cdot(n-1)+ (1+\eps) \cdot( \sqrt{w(n-2)}+\sqrt{w}) \\
	&c_n^2 -c_n\cdot \left( \frac{e}{e-1} \cdot(n-1)+ (1+\eps) \cdot( \sqrt{w(n-2)}+\sqrt{w})\right) -\frac{w(1+\varepsilon)}{2} \leq 0
	\end{align*}
	We have a quadratic inequality in $c_n$ and we wish to find the maximal value of $c_n$ for which this inequality holds. Denote by $b=( \frac{e}{e-1} \cdot(n-1)+ (1+\eps) \cdot( \sqrt{w(n-2)}+\sqrt{w}))$ and $c=\frac{w(1+\varepsilon)}{2}$. The following argument uses the quadratic formula to show that: $c_n \leq  b + \frac{c}{b}$: 
	
	\begin{align*}
	c_n \leq \frac{b+ \sqrt{b^2+4c}}{2} = \frac{b}{2} \cdot \left(1+ \sqrt{\left(1+\frac{4c}{b^2}\right)}\right) < \frac{b}{2} \cdot \left(2+\frac{2ac}{b^2}\right) = b \cdot \left(1+\frac{ac}{b^2}\right) = b + \frac{c}{b}
	\end{align*}
	
	Now by substituting $b$ and $c$ for their explicit values and simplifying we get that:
	\begin{align*}
	c_n &\leq \frac{e}{e-1} \cdot(n-1)+ (1+\eps) \cdot( \sqrt{w(n-2)}+\sqrt{w}) + \frac{w(1+\eps)}{2(\frac{e}{e-1} \cdot(n-1)+ (1+\eps) \cdot( \sqrt{w(n-2)}+\sqrt{w}))} \\
	&\leq \frac{e}{e-1} \cdot(n-1)+ (1+\eps) \cdot( \sqrt{w(n-2)}+\sqrt{w}) + \frac{w(1+\eps)}{ 2 (1+\eps)(\sqrt{w(n-2)}+\sqrt{w})} \\
	&\leq \frac{e}{e-1} \cdot(n-1)+ (1+\eps) \cdot( \sqrt{w(n-2)}+\sqrt{w}) + \frac{w}{ 2(\sqrt{w(n-2)}+\sqrt{w})} \\
	&\leq \frac{e}{e-1} \cdot(n-1)+ (1+\eps) \cdot( \sqrt{w(n-1)}+\sqrt{w}).
	\end{align*}
	The last inequality holds since $\sqrt{n-2}+1 +\frac{1}{2(\sqrt{n-2}+1)} \leq 1 + \sqrt{n-2}$ (the simple proof can be found in auxiliary Lemma \ref{lem-sqrt} in the appendix).
\end{proof}

Lastly, we note that by taking $w>\frac{2n}{\eps^2}$ such that it holds that $\frac{e}{e-1}  \cdot n\leq \eps\sqrt{w \cdot n}$ we get the following corollary: 
\begin{corollary} \label{cor:ub-large-w}	
	Fix $n\geq 2$. For any $0<\eps\leq 1$, there exists $w_1^n(\eps)$ such that for any $w \geq w_1^n(\eps)$ it holds that for every symmetric equilibrium $S$:	
	$c_{n,w}(S) \leq (1+2\eps) \cdot( \sqrt{w} \cdot \sqrt{n+2\sqrt{n-1}})$
	and thus the social cost 
	satisfies  $C_{n,w}(S)\leq (1+2\eps) \cdot( \sqrt{w} \cdot \sqrt{n+2\sqrt{n-1}})\cdot n$. 
\end{corollary}
\begin{proof}
	By the second statement of  Proposition \ref{prop:nash-upperbound-large}, 
	for every  $n\geq 2$ and for any $0<\eps\leq 1$, 
	there exists $w_0^n(\eps)$ such that for any $w \geq w_0^n(\eps)$: 
	\begin{align*}
	c_{n}=c(n,0) \leq \frac{e}{e-1}  \cdot n + (1+\eps) \cdot \left( \sqrt{w} \cdot\sqrt{n+2\sqrt{n-1}}\right)
	\end{align*}
	For $w>\frac{2n}{\eps^2}$ it holds that $\frac{e}{e-1}  \cdot n\leq \eps\sqrt{w \cdot n}$ and 
	thus, for $w_1^n(\eps) = \max \{w_0^n(\eps), \frac{2n}{\eps^2}\}$ it holds that 
	\begin{align*}
	c_{n} \leq  (1+2\eps) \cdot \left( \sqrt{w} \cdot \sqrt{n+2\sqrt{n-1}}\right)
	\end{align*}
\end{proof}

\subsubsection{A Tight Lower Bound for Large $w$} \label{sec:large-w-lower-bound}
We next present a tight lower bound that holds for large values of $w$ and only grows asymptotically in $w$ as $\sqrt{w}$. The proof uses the upper bound from Proposition \ref{prop:nash-upperbound-large}. 
\begin{proposition} \label{prop:nash-lowerbound}
For every  $n\geq 2$ and every $0<\eps < 1/2$, there exists $w_l^n(\eps)$ such that for any $w \geq w_l^n(\eps)$ in every symmetric equilibrium: 
\begin{align*}
c_{n} \geq (1-2\eps)^{n-1} \cdot \frac{\sqrt{w}}{2} \cdot  \sum_{i=1}^{n-1} \left(\frac{1} {1+\sqrt{i}} \right)
\end{align*}

\end{proposition}
\begin{proof}
	Consider any symmetric equilibrium according to Observation \ref{obs:eq-costs-0} we have that:
	\begin{align*}
	c_n &= \frac{1}{1-(1-q_n)^{n-1}} + \frac{1}{1-(1-q_n)^{n-1}} \cdot \sum_{i=1}^{n-1} {{n-1}\choose{i}} q_n^i \cdot (1-q_n)^{n-1-i} c(n-i,i-1)
	\end{align*}
	Thus, to get a lower bound we can only consider the first term in the sum and get that:
	\begin{align*}
	c_n &\geq \frac{1}{1-(1-q_n)^{n-1}} + \frac{(n-1) q_n \cdot (1-q_n)^{n-2}}{1-(1-q_n)^{n-1}}  \cdot c_{n-1} 
	\end{align*}
	Observe that $\frac{1}{1-(1-q_n)^{n-1}} \geq \frac{1}{(n-1)q_n}$ (a simple proof for this can be found in Lemma \ref{lem:bounds-when-p-small} in the appendix.

	Also, $\frac{1}{(n-1)q_n}= \frac{w}{2c_n}$ since $c_n = c_n^1(q_n) =  \frac{n-1}{2} \cdot q_n \cdot w$. Thus, we have:

	\begin{align*}
	c_n  &\geq \frac{w}{2c_n} + \frac{(n-1) q_n \cdot (1-q_n)^{n-2}}{q_n (n-1)}  \cdot c_{n-1} \\ 
	&\geq \frac{w}{2c_n} + \left(1- q_n(n-2) \right)\cdot c_{n-1} 
	\end{align*}

	By Proposition \ref{prp:nash:zero} we have that for every $n\geq 2$ and $0<\eps \leq 1$ there exists $w_0^n(\eps)$ such that for every $w\geq w_0^n(\eps)$, we have that  $q_n \cdot (n-1) \leq \eps$.
	\begin{align*}
	c_n &\geq \frac{w}{2c_n} + \left(1- q_n(n-2) \right)\cdot c_{n-1} \geq \frac{w}{2c_n} + (1-\eps) \cdot c_{n-1}
	\end{align*}
	Next, we can plug in the upper bound on $c_n$ from Proposition \ref{prop:nash-upperbound-large}  
	and get that for any $w\geq w_0^n(\eps)$:
	\begin{align*}
	c_n &\geq \frac{w}{2\cdot \left(\frac{e}{e-1}  \cdot n + (1+\eps) \cdot\left( \sqrt{w} \cdot \sqrt{n+2\sqrt{n-1}}\right)\right)} + (1-\eps) \cdot c_{n-1}
	\end{align*}
	
	Pick $w_l^n(\eps) = \max \{ w_0^n(\eps), \frac{4n}{\eps^2} \}$. 
	For $w>w_l^n(\eps)$ it holds that 
	
	\begin{align*}
	\eps \cdot \sqrt{w} \cdot \sqrt{n+2\sqrt{n-1}} 
	&\geq \eps \cdot \sqrt{\frac{4n}{\eps^2}} \cdot \sqrt{n+2\sqrt{n-1}} 
	> 2n > \frac{e}{e-1}  \cdot n
	\end{align*}
	
	Thus, using the fact that $\sqrt{w} \cdot \sqrt{n+2\sqrt{n-1}}  =  \sqrt{w(n-1)}+\sqrt{w}$, we have that: 
	
	\begin{align*}
	c_n &\geq \frac{w}{2\cdot \left((1+2\eps) \cdot \left( \sqrt{w(n-1)}+\sqrt{w}\right)\right)} + (1-\eps) \cdot c_{n-1} \\
	&\geq (1-2\eps) \cdot \frac{\sqrt w}{2\cdot\left( \sqrt{(n-1)}+1\right)} + (1-\eps) \cdot c_{n-1} \\
	&\geq (1-2\eps) \cdot \left(\frac{1}{2} \frac{\sqrt{w}}{\sqrt{n-1}+1}+ c_{n-1} \right)
	\end{align*}
	
	Note that the above holds for any $n' \leq n$ and $w > w_l^n(\eps)$, hence, by induction we get that:
	\begin{align*}
	c_n &\geq \sum_{i=1}^{n-1} (1-2\eps)^i \left( \frac{1}{2} \frac{\sqrt{w}}{\sqrt{n-i}+1} \right) \\
	&\geq (1-2\eps)^{n-1} \cdot  \frac{\sqrt{w}}{2} \cdot  \sum_{i=1}^{n-1} \left(\frac{1} {\sqrt{n-i}+1} \right) \\
	&=(1-2\eps)^{n-1} \cdot \frac{\sqrt{w}}{2} \cdot  \sum_{i=1}^{n-1} \left(\frac{1} {1+\sqrt{i}} \right)
	\end{align*}
	
	where the second inequality follows from Lemma \ref{lem:bounds-when-p-small}. 
\end{proof}

We use the proposition to derive the following corollary lower bounding the asymptotic grows of the cost when both $n$ is large, and $w$ is also large enough.  

\begin{corollary}\label{cor:lb-large-w}
	Fix $n\geq 2$. For every $0<\eps<1/2$, there exists $w_0^n(\eps)$ such that for any 
	$w \geq w_0^n(\eps)$ in every symmetric equilibrium: 
	$$c_n \geq (1-2\eps)^{n-1} \cdot \sqrt{w(n-2\ln n)}$$
\end{corollary}

\begin{proof}
By Proposition \ref{prop:nash-lowerbound} for every fixed $n\geq 2$ and $0<\eps < 1/2$, there exists $w_0^n(\eps)$ such that for any $w \geq w_0^n(\eps)$: 
\begin{align*}
c_n &\geq (1-2\eps)^{n-1} \cdot \frac{\sqrt{w}}{2} \cdot  \sum_{i=1}^{n-1} \left(\frac{1} {1+\sqrt{i}} \right)
\end{align*}

It is not hard to show that 
$$\sum_{i=1}^{n-1} \left(\frac{1} {1+\sqrt{i}} \right) \geq 2\left(\sqrt{n}- \ln (1+\sqrt{n})-1 + \ln 2 \right)$$
(we prove this in Lemma \ref{lem:sum-sqrt-rec} in the appendix). This implies that:
\begin{align*}
c_n &\geq (1-2\eps)^{n-1} \cdot \sqrt{w} \left( \sqrt{n} - \ln (1+\sqrt{n}) -1  + \ln 2 \right) 
\end{align*}

Next, we observe that for $n\geq 2$
\begin{align*}
\ln (1+\sqrt{n}) +1  - \ln 2 \leq \ln (2\sqrt{n}) +1  - \ln 2 = \ln 2 + \frac{1}{2}\ln n +1  - \ln 2= \frac{1}{2}\ln n +1  < 2\ln{n}
\end{align*}

Thus, we have that $c_n \geq (1-2\eps)^{n-1} \cdot \sqrt{w(n-2\ln n)}$ as required.
\end{proof}	
}
  
\section{Bounds on the Cost of Symmetric Optimal Solutions}
\label{sec:symmetric-OPT}
In this section we provide bounds on the {\em optimal symmetric cost}, this is the minimal cost when agents are restricted to play symmetric anonymous stationary strategies, but not necessarily equilibrium strategies. That is, for $w>1$ and $n\geq 2$, the optimal cost $OPT(n,w) = \inf_S C_{n,w}(S)$, the infimum is taken over all profiles $S$ of symmetric profiles of anonymous stationary strategies. 

It is easy to see that since the cost of waiting in the queue is greater than the cost of waiting outside ($w>1$) it is socially suboptimal to direct an agent to enter when the queue is not empty. Thus, when considering an optimal symmetric strategy in the game $G(n;w)$ we can restrict ourselves to strategies that define an entrance probability $p_m$ for every number of players $1\leq m \leq n$ such that the queue is empty and there are $m$ agents outside the queue.

Providing a closed form expression for the optimal symmetric cost for the game $G(n;w)$ is very challenging as it requires minimizing a function which is the ratio of two polynomials each of degree $n-1$. Hence, we compute lower and upper bounds on the optimal symmetric cost instead. As in the case of the symmetric equilibrium we provide different bounds for the case that  $w$ is fixed and $n$ goes to infinity and for the case that $n$ is fixed and $w$ goes to infinity:

\begin{theorem}\label{thm:opt}
	$OPT(n,w)$ is bounded as follows:
	\begin{itemize}
			\item 	Fix any $w>2$ 
			and any $\eps>0$. 
			There exists $n_0^w(\eps)$ such that for any $n>n_0^w(\eps)$ it holds that   
			$$ 
			\frac{n(n-1)}{2}\leq OPT(n,w) \leq (1+\eps) \frac{n(n-1)}{2}
			$$
		\item 
		Fix any $n\geq 2$ and any $\eps>0$. 
		There exists $w_1^n(\eps)$ such that for any $w>w_1^n(\eps)$ it holds that		
		$$
				(1-\eps) \sqrt{2w} \cdot \frac{2}{3} (n-1)\sqrt{n-1} \leq OPT(n,w) \leq (1+\eps) \sqrt{2w} \cdot \left(\frac{2}{3} n\sqrt{n}+\sqrt{n}\right)
		$$
	\end{itemize}	
\end{theorem}

From the theorem we derive two corollaries about the asymptotic cost of $OPT(n,w)$.
Our first corollary shows that for fixed $w>2$, the symmetric optimal cost grows asymptotically the same as the cost of the optimal schedule, when $n$ grows large. This implies that in this case the source of the inefficiency of the symmetric equilibrium is {\em only} due to strategic behavior, and not lack of coordination.

\begin{corollary}
For every fixed $w>2$ and every $\eps>0$ there exists $n_0^w(\eps)$ such that for any $n>n_0^w(\eps)$ it holds that
$1\leq OPT(n,w)/SC(n,w)\leq 1+\eps$.
\end{corollary} 

Combining the previous corollary together with Corollary \ref{cor:nash-2} establishes the proof of Theorem \ref{thm:intro-symmetric-POA-fixed-w}. Thus, we have that for every fixed $w>2$ and every $\eps>0$ there exists $n_0^w(\eps)$ such that for any $n>n_0^w(\eps)$ it holds that
$2-\eps\leq C_{n,w}(S)/OPT(n,w)\leq 2+\eps$.

For the case that $n$ is large enough and $w$ is greater than $n$ we get that the cost of a symmetric optimal solution is about $\frac{2\sqrt 2}{3} \cdot n \cdot \sqrt{w\cdot n}$. Recall that for the same case by Theorem \ref{thm:opt-large-w} we have that the social cost of any symmetric Nash equilibrium is about $n \cdot \sqrt{w\cdot n}$. Thus, we have that in this case the ratio between the cost of any symmetric equilibrium and the symmetric optimal cost  converges to $\frac{3}{2\sqrt{2}} \approx  1.06$, which proves Theorem \ref{thm:intro-symmetric-POA-fixed-n}. This means that in the case that $w$ is relatively larger than $n$ the lack of coordination is playing a major role in deteriorating the efficiency of symmetric equilibria. This explains why in such cases, it is common that measures are taken to increase coordination and help boost social welfare. 
Formally, In Appendix \ref{app-cor-opt} we show that:

\begin{corollary}\label{cor:ratio-opt-sw}
	Fix any $\delta>0$. 
	There exists $n_0(\delta)$ such that for any $n>n_0(\delta)$ there exist $w^{n}(\delta)$ such that for any $w>w^{n}(\delta)$ it holds that for any symmetric equilibrium $S$   
		$$ 
		(1-\delta)
		\frac{3}{2\sqrt{2}}  \leq \frac{C_{n,w}(S)}{OPT(n,w)} \leq 
		(1+\delta) \frac{3}{2\sqrt{2}} 
		$$
	We note that $\frac{3}{2\sqrt{2}}\approx 1.061 $.		 	
\end{corollary} 

The starting point for our proof of Theorem \ref{thm:opt} is the following recursive formula for the cost of the optimal solution:

\short{
\begin{align*}
OPT(n,w)= (1-p_n)^n (n+OPT(n,w)) + \sum_{i=1}^{n} p_n^i (1-p_n)^{n-i} {n \choose i} \left(w\sum_{j=1}^{i-1} j+ i(n-i) + OPT(n-i,w)\right)
\end{align*}
}

\full{
\begin{align} \label{eq:optimum}
OPT(n,w)= (1-p_n)^n (n+OPT(n,w)) + \sum_{i=1}^{n} p_n^i (1-p_n)^{n-i} {n \choose i} \left(w\sum_{j=1}^{i-1} j+ i(n-i) + OPT(n-i,w)\right)
\end{align}
}

The base case for the formula is $OPT(1,w)=0$ for any $w$ as the cost of an agent that is serviced right away is $0$. Our recursive formula holds for any $w>1$ and $n\geq 2$. To compute the expected cost the recursive function splits the cost by the number of agents that enter the queue.
With probability $(1-p_n)^n$ no agent will enter the queue, and each will pay a cost of $1$ (total of $n$), and we will be back at the same situation with $n$ agents, so the future cost will be $OPT(n,w)$. 
Additionally, for each $i\in [n]$ we consider the case that exactly $i$ agents enter the queue.
The probability that exactly $i$ will enter the queue is  $p_n^i (1-p_n)^{n-i} {n \choose i}$. 
In any such case, we compute the cost till all these $i$ agents are processed: the total cost for the $i$ agents that enter the queue is $w\sum_{j=1}^{i-1} j$ (in the first time unit $i-1$ agents pay $w$ and one agent is processed, then in the second time unit $i-1$ agents pay $w$ an one agent is processed, till the $i$ time step where the last agent in the queue is processed) and the total cost for the $n-i$ agents that did not enter the queue is $i(n-i)$ as each waits $i$ steps. After $i$ steps the queue is empty again, and the total remaining cost is $OPT(n-i,w)$.   

For both of the upper bounds we first simplify the recursive formula to get an upper-bound which depends only on $OPT(n-1,w)$. To this end we use the observation that if there are $n-2$ agents the additional cost of processing another agent in optimal symmetric strategies is at least $n-2$. This translates to 
the upper-bound $OPT(n-2,w) \leq OPT(n-1,w) -(n-2)$, which can be easily extended to upper bound $OPT(n-i,w)$ for any $i$. The optimal symmetric cost is clearly bounded by the cost that can be achieved by any symmetric strategies, and thus we can guess different entrance probabilities for the different parameter regimes and bound their cost. For the case that $w$ is fixed and $n$ is large we use $p_n=\frac{\log n}{n} \cdot \sqrt{\frac{2}{w}}$ while for the case that $w$ is much larger than $n$ we use $p_n = \frac{1}{n} \sqrt{\frac{2(n-1)}{w-1}}$. The key idea here is basically balancing out the number of steps that we need to wait till a single agent joins the queue with the expected cost in case multiple agents enter the queue simultaneously and a cost of $w$ is exhibited. 

As for the lower bounds. The lower bound for the case of a fixed $w$ and an increasing $n$ is simply the cost of the unrestricted optimal solution in which all players enter sequentially. The lower bound for the case that $w$ is large is based on the recursive formula for $OPT(n,w)$ as well. Now to get a lower bound which is only a function of $OPT(n-1,w)$ we compute a lower bound on the cost of the case that two agents join the queue simultaneously and simply ignore all cases in which more agents simultaneously join the queue. This approach produces a tight lower bound since in the case that $w$ is large relatively to $n$, the entrance probability is so low that the event that more than two agents join the queue is highly unlikely.  

While the above intuition captures some of the essence of our proofs, formalizing it is not so easy. \short{The complete proofs can be found in Appendix \ref{app:symmetric-OPT}. In particular,the proofs for the upper bounds and lower bounds can be found in Section \ref{sec:OPT-UB} and Section \ref{sec:OPT-LB} respectively. } \full{In the next sections we prove Theorem \ref{thm:opt}. In particular, the proofs for the upper bounds and lower bounds can be found in Section \ref{sec:OPT-UB} and Section \ref{sec:OPT-LB} respectively.} 
\full{ 
	

\short{
In this section we will use the recursive formula for the expected cost of the optimal solution from Section \ref{sec:symmetric-OPT} to bound the cost of the optimal symmetric solution. 
Recall that the recursive formula states that:
\begin{align}
\label{eq:optimum}
OPT(n,w)= (1-p_n)^n (n+OPT(n,w)) + \sum_{i=1}^{n} p_n^i (1-p_n)^{n-i} {n \choose i} \left(w\sum_{j=1}^{i-1} j+ i(n-i) + OPT(n-i,w)\right)
\end{align}
}

\subsection{Upper Bounds on $OPT(n,w)$ } \label{sec:OPT-UB}
\subsubsection{A General Upper Bound which is Only a Function of $OPT(n-1)$}
Our goal is to upper bound $OPT(n,w)$, towards this goal, our first step is simplifying the recursive formula of Equation (\ref{eq:optimum}) to get an upper bound that is {\em only} a function of $OPT(n-1)$:

\begin{lemma}\label{lem:opt-ub}
		Fix any $w>2$ and $n\geq 2$. Let $p=p_n$ and $\alpha = p\cdot n$.
		It holds that:
		\begin{equation}
		\label{eq:opt-ub}
		OPT(n)  \leq  OPT(n-1)  + n-1 + \frac{n e^{-\alpha} + \frac{w-1}{2} \cdot \alpha^2 \cdot \left(\frac{n}{n-\alpha}\right)^2 \cdot  \exp\left(\frac{\alpha^2}{n-\alpha}\right)}{1-e^{-\alpha}} 
		\end{equation}
\end{lemma}	
\begin{proof}
First, we denote the second term in the recursive formula of $OPT(n)$ by $X$, that is:
\begin{align*}
OPT(n)= (1-p)^n (n+OPT(n)) + \underbrace{\sum_{i=1}^{n} p^i (1-p)^{n-i} {n \choose i} \left(w\sum_{j=1}^{i-1} j+ i(n-i) + OPT(n-i)\right)}_X
\end{align*}
Using the notation of $X$ we have that $OPT(n) =(1-p)^n (n+OPT(n)) + X$. After some rearranging we get that: 
\begin{align*}
OPT(n) = \frac{(1-p)^n \cdot n + X}{1-(1-p)^n}
\end{align*}

Our next step is computing a bound on $X$ which is only a function of $OPT(n-1)$. Towards this end, we observe that for any $m>1$ it holds that $OPT(m)\geq OPT(m-1) + m-1$ as at best, when we add an agent, it will take $OPT(m-1)\geq m-1$ to process all agents but him, so he will need to wait at least $m-1$ steps, and this will be added to the cost of $OPT(m-1)$. 
Using this fact, we prove the next lemma by induction.
	\begin{claim} \label{clm:opt-simple-up}
	For any $n\geq 2$ and $i\geq 1$ it holds that $OPT(n-i) \leq OPT(n-1)- (i-1)n + \frac{(i-1) (i+2)}{2}$.
	\end{claim}
	
	\begin{proof}
		Recall that $OPT(m)\geq OPT(m-1) + m-1$, by induction on $m$ this implies that for any $i'\geq 0$:
		\begin{align*}
		OPT(m) \geq OPT(m-i') + \sum_{j=1}^{i'}(m-j) = OPT(m-i') +  i'\cdot m - \frac{i'(i'+1)}{2}
		\end{align*}
		By rearranging we get that for $m=n-1$ and any $i'\geq 0$:
		\begin{align*}
		OPT(n-1-i')\leq OPT(n-1)- i'(n-1) + \frac{i'(i'+1)}{2}
		\end{align*}
		Now by setting $i = i'+1$ we get that for every $i \geq 1$:
		\begin{align*}
		OPT(n-i) & \leq OPT(n-1)- (i-1)(n-1) + \frac{(i-1)i}{2} \\
		&\leq OPT(n-1)- (i-1)n + \frac{(i-1) (i+2)}{2}
		\end{align*}
		as required.
	\end{proof}

Recall that $X = \sum_{i=1}^{n} p^i (1-p)^{n-i} {n \choose i} \left(w\sum_{j=1}^{i-1} j+ i(n-i) + OPT(n-i)\right)$.
We next use the bound on $OPT(n-i)$  presented in Claim \ref{clm:opt-simple-up} to upper bound $X$ by some function of $OPT(n-1)$:
\begin{align*}
X & \leq \sum_{i=1}^{n} p^i (1-p)^{n-i} {n \choose i} \left(w\frac{i(i-1)}{2}+ i(n-i) + OPT(n-1)- n(i-1) + \frac{(i-1) (i+2)}{2}\right) \\ 
&= (1-(1-p)^n) \cdot OPT(n-1)  + \sum_{i=1}^{n} p^i (1-p)^{n-i} {n \choose i} \left(w\frac{i(i-1)}{2} + n-1 - \frac{i(i-1)}{2}\right) \\
&= (1-(1-p)^n) \cdot \left( OPT(n-1) + n-1 \right)  + \sum_{i=1}^{n} p^i (1-p)^{n-i} {n \choose i} \left(w\frac{i(i-1)}{2} - \frac{i(i-1)}{2}\right) 
\end{align*}
	
By substituting $X$ with its upper bound and rearranging we get that:  

	\begin{align*}
	OPT(n) &\leq   OPT(n-1) +n-1 + \frac{n(1-p)^n + \sum_{i=1}^{n} p^i (1-p)^{n-i} {n \choose i} \left(w\frac{i(i-1)}{2} - \frac{i(i-1)}{2}\right)}{1-(1-p)^n}  \\\
	& \leq OPT(n-1)  + n-1 + \frac{n (1-p)^n + \frac{w-1}{2} \sum_{i=2}^{n} p^i (1-p)^{n-i} {n \choose i} \left(i(i-1)\right)}{1-(1-p)^n} 
	\end{align*}
	
	Recall that $p=\frac{\alpha}{n}$. Using this notation we have that:

		\begin{align*}
		OPT(n) \leq OPT(n-1)  + n-1 + \frac{n (1-\frac{\alpha}{n})^n + \frac{w-1}{2} \sum_{i=2}^{n} (\frac{\alpha}{n})^i (1-\frac{\alpha}{n})^{n-i} {n \choose i} \left(i(i-1)\right)}{1-(1-\frac{\alpha}{n})^n} 
		\end{align*}

	Observe that $(1-\frac{\alpha}{n})^n\leq e^{-\alpha}$ and that for $\beta>0$ by the Taylor series of $e^\beta$
	it holds that $\sum_{i=2}^{n} \frac{\beta^i}{(i-2)!}\leq e^\beta \beta^2$. 
	We derive: 
	
	\begin{align*}
	\sum_{i=2}^{n} \left(\frac{\alpha}{n}\right)^i \left(1-\frac{\alpha}{n}\right)^{n-i} {n \choose i} \left(i(i-1)\right)& =
	\left(1-\frac{\alpha}{n}\right)^{n} \sum_{i=2}^{n} \frac{\alpha^i}{(n-\alpha)^i}  \frac{n!}{(n-i)! i!} \left(i(i-1)\right) \\ & =
	\left(1-\frac{\alpha}{n}\right)^{n} \sum_{i=2}^{n}  \frac{\alpha^i}{(i-2)!} \cdot  \frac{n!}{(n-\alpha)^i (n-i)! }  \\  & 
	\leq e^{-\alpha} \sum_{i=2}^{n}  \frac{1}{(i-2)!} \cdot  \left(\frac{n\alpha}{n-\alpha}\right)^i  \\ & \leq 
	e^{-\alpha}  \left(\frac{n\alpha}{n-\alpha}\right)^2 \cdot   \exp\left(\frac{n\alpha}{n-\alpha}\right) \\ & = 
	\left(\frac{n\alpha}{n-\alpha}\right)^2 \cdot  \exp\left(\frac{\alpha^2}{n-\alpha}\right) 
	\end{align*}

	Thus: 
	\begin{equation*}
	OPT(n)  \leq  OPT(n-1)  + n-1 + \frac{n e^{-\alpha} + \frac{w-1}{2} \cdot \alpha^2 \cdot \left(\frac{n}{n-\alpha}\right)^2 \cdot  \exp\left(\frac{\alpha^2}{n-\alpha}\right)}{1-e^{-\alpha}} 
	\end{equation*}
	as needed.
\end{proof}

\subsubsection{Two Concrete Upper Bounds}

We present two bounds, one when $w$ is large ($n$ is fixed and $w$ grows to infinity), and the other when $w$ is small ($w$ is fixed and $n$ grows to infinity). Both bounds are derived from the upper bound in Lemma \ref{lem:opt-ub}. For the case of fixed $w$ and a large $n$ we use $p=\frac{\log m}{m} \cdot \sqrt{\frac{2}{w}}$ that is, $\alpha = \sqrt{\frac{2}{w}} \log m $, 
and show in Proposition \ref{prop:opt-ub-large-n} that for any $w$, as $n$ grows large the optimal cost with symmetric strategies converges to the unrestricted optimal cost.
After that, In Proposition \ref{prop:opt-ub-large-w} we derive an upper bound for a fixed $n$ and large $w$ by using $p = \frac{1}{m} \sqrt{\frac{2(m-1)}{w-1}}$. 
	
\begin{proposition}\label{prop:opt-ub-large-n}
Fix any $w>2$ and any $0< \eps < 1$. There exists $n_0^w(\eps)$ such that for any $n>n_0^w(\eps)$ it holds that   
$$ 
OPT(n,w) \leq (1+\eps) \frac{n(n-1)}{2}
$$
\end{proposition} 	
\begin{proof}
Recall that for any $m\leq n$, if we define $p=p_m$ and $\alpha = p \cdot m$ then by Lemma \ref{lem:opt-ub} we have 
\begin{equation*}
OPT(m)  \leq  OPT(m-1)  + m-1+ \frac{m e^{-\alpha} + \frac{w-1}{2} \cdot \alpha^2 \cdot \left(\frac{m}{m-\alpha}\right)^2 \cdot  \exp\left(\frac{\alpha^2}{m-\alpha}\right)}{1-e^{-\alpha}} 
\end{equation*}
		
Fix $p=\frac{\log m}{m} \cdot \sqrt{\frac{2}{w}}$ this implies that $\alpha = \sqrt{\frac{2}{w}} \log m $. For such $\alpha$ we get that 
$e^{\alpha} = e^{\sqrt{\frac{2}{w}} \log m} = m^{\sqrt{\frac{2}{w}}}$ and thus
$1-e^{-\alpha} = 1- m^{-\sqrt{\frac{2}{w}}}$ and this term tends to $1$ as $m$ grows to infinity.
For large enough $m$, the term $m-\alpha$ tends to $m$ and $ \left(\frac{m}{m-\alpha}\right)^2$ tends to $1$. 
Additionally, $\alpha^2 = \frac{2}{w} (\log m)^2$. 
The term $\exp\left(\frac{\alpha^2}{m-\alpha}\right)$ 
equals to $\exp\left(\frac{ \frac{2}{w} (\log m)^2}{m- \sqrt{\frac{2}{w}} \log m}\right) $
which goes to $1$ from above as $m$ goes to infinity, since $\frac{ \frac{2}{w} (\log m)^2}{m- \sqrt{\frac{2}{w}} \log m} = {\left(\frac{2\log^2 m}{m\cdot w - \sqrt{2w} \log m}\right)}$ tends to $0$ from above as $m$ goes to infinity.
We get that when $m$ grows the term $\frac{w-1}{2} \cdot \alpha^2 \cdot \left(\frac{m}{m-\alpha}\right)^2 \cdot  \exp\left(\frac{\alpha^2}{m-\alpha}\right)$ tends to 
$\frac{w-1}{2} \cdot  \frac{2}{w} (\log m)^2 \leq   (\log m)^2$.

Thus, for any $w>2$ and $\eps_1,\eps_2>0$, for large enough $m$ it holds that: 

\begin{align*}
OPT(m)  \leq  OPT(m-1)  + m-1+ \frac{ \overbrace{m e^{-\alpha}}^{\text{= $m^{\left(1-\sqrt{\frac{2}{w}}\right)}$}} + \overbrace{\frac{w-1}{2} \cdot \alpha^2}^{=\frac{w-1}{w} (\log m)^2} \cdot \overbrace{\left(\frac{m}{m-\alpha}\right)^2 \cdot  \exp\left(\frac{\alpha^2}{m-\alpha}\right)}^{\text{is at most $1+\eps_1$}}}{\underbrace{1-e^{-\alpha}}_{\text{is at least $1-\eps_2$}}} 
\end{align*}

or alternatively, 

\begin{align*}
OPT(m)  & \leq  OPT(m-1)  + m-1+ \frac{ m^{\left(1-\sqrt{\frac{2}{w}}\right)} + \frac{w-1}{w} (\log m)^2 \cdot (1+\eps_1)}{1-\eps_2} \\
& \leq  OPT(m-1)  + m-1+ \frac{1+\eps_1}{1-\eps_2} \left( m^{\left(1-\sqrt{\frac{2}{w}}\right)} + \frac{w-1}{w} (\log m)^2 \right)
\end{align*}

Thus, for large enough $m$, for $\delta=\frac{1+\eps_1}{1-\eps_2}>1$ it holds that 
\begin{equation*}
OPT(m)  \leq  OPT(m-1)  + m-1+\delta \left(m^{\left(1-\sqrt{\frac{2}{w}}\right)}+ (\log m)^2\right)
\end{equation*}	
By induction, and monotonicity of $\left(m^{\left(1-\sqrt{\frac{2}{w}}\right)}+ (\log m)^2\right)$ in $m$,  it holds that 
\begin{equation*}
OPT(n)  \leq  \sum_{m=1}^{n} \left(m-1+\delta \left(m^{\left(1-\sqrt{\frac{2}{w}}\right)}+ (\log m)^2\right)\right) \leq  \frac{n(n-1)}{2} + \delta \cdot n \left(n^{\left(1-\sqrt{\frac{2}{w}}\right)}+ (\log n)^2\right)
\end{equation*}	

As $\delta \left(n^{\left(1-\sqrt{\frac{2}{w}}\right)}+ (\log n)^2\right)$ grows asymptotically much slower than $(n-1)/2$
(since for $w>2$ the function $n^{\left(1-\sqrt{\frac{2}{w}}\right)}$ is asymptotically smaller than $n$), it holds that  for any $0<\eps<1$, for $w>2$ and for large enough $n$: 
		
\begin{equation*}
OPT(n) \leq (1+\eps) \frac{n(n-1)}{2}
\end{equation*}

\end{proof}

Next, we derive an upper bound for a fixed $n$ and large $w$ by using $p = \frac{1}{m} \sqrt{\frac{2(m-1)}{w-1}}$. 
\begin{proposition}\label{prop:opt-ub-large-w}
Fix any $n\geq 2$ and any $0<\eps<1$. There exists $w_0^n(\eps)$ such that for any $w>w_0^n(\eps)$ it holds that   

\begin{equation*}
OPT(n,w) \leq (1+\eps) \sqrt{2w} \cdot \sum_{i=1}^{n} \sqrt{i} 
\leq (1+\eps) \sqrt{2w} \cdot \left(\frac{2}{3} n\sqrt{n}+\sqrt{n}\right)
\end{equation*}

\end{proposition} 	
\begin{proof}
The right inequality follows from the  Euler--Mclaurin formula : $\sum_{i=1}^{k} \sqrt{i} < \frac{2}{3} k\sqrt{k}+\sqrt{k}$, so we next prove the left inequality. Recall that By Lemma \ref{lem:opt-ub} we have that for any $m\leq n$
\begin{equation*}
OPT(m)  \leq  OPT(m-1)  + m-1+ \frac{m e^{-\alpha} + \frac{w-1}{2} \cdot \alpha^2 \cdot \left(\frac{m}{m-\alpha}\right)^2 \cdot  \exp\left(\frac{\alpha^2}{m-\alpha}\right)}{1-e^{-\alpha}} 
\end{equation*}
			
Fix $p = \frac{1}{m} \sqrt{\frac{2(m-1)}{w-1}}$ this implies that
$\alpha = \sqrt{\frac{2(m-1)}{w-1}}$. Note that $\alpha$ tends to $0$ as $w$ grows to infinity. Additionally, $\alpha$ is positive and thus $e^{-\alpha}\leq 1$. Fix any $\delta_1>0$.
The term $\left(\frac{m}{m-\alpha}\right)^2 \cdot  \exp\left(\frac{\alpha^2}{m-\alpha}\right)$
tends to $1$ as $w$ grows to infinity since $\alpha$ tends to $0$, so for any $m\leq n$, for large enough $w$, this term is smaller than $1+\delta_1$.
Additionally, $1-e^{-\alpha}$ tends to $\alpha$ from below as $w$ grows to infinity, so for large enough $w$ it is greater than $(1-\delta_1) \alpha$. Thus,  for any fixed $n$ and any $m\leq n$, for large enough $w$ we have 
\begin{align*}
OPT(m)  \leq  OPT(m-1)  + m-1+ \frac{ \overbrace{m e^{-\alpha}}^{\text{is at most $m$}} + \overbrace{\frac{w-1}{2} \cdot \alpha^2}^{=m-1} \cdot \overbrace{\left(\frac{m}{m-\alpha}\right)^2 \cdot  \exp\left(\frac{\alpha^2}{m-\alpha}\right)}^{\text{is at most $1+\delta_1$}}}{\underbrace{1-e^{-\alpha}}_{\text{is at least $(1-\delta_1)\alpha$}}} 
\end{align*}

Therefore, for $\delta = \frac{1+\delta_1}{1-\delta_1}$ we have that:
\begin{align*}
OPT(m) & \leq OPT(m-1) + m-1 + \delta \cdot \frac{2m-1}{\alpha} \\ &=  OPT(m-1) + m-1 + \delta \cdot \left(\frac{2(m-1)}{\alpha} + \frac{1}{\alpha} \right)  \\ &  = 
OPT(m-1) + m-1+\delta \cdot \left(\sqrt{2(m-1)(w-1)} +\sqrt{\frac{w-1}{2(m-1)}}\right) \\
&  = 
OPT(m-1) + m-1+\delta \cdot \sqrt{2(w-1)} \left(\sqrt{m-1} +\frac{1}{2\sqrt{m-1}}\right)
\end{align*}
Therefore, for large enough $w$, $OPT(n)$ is bounded from above as follows: 
\begin{align*}
OPT(n) \leq \frac{n(n-1)}{2}+ \delta\cdot \sqrt{2(w-1)} \cdot \left( \sum_{i=1}^{n-1} \sqrt{i} + \frac{1}{2}\sum_{i=1}^{n-1} \frac{1}{\sqrt{i}} \right)
\end{align*}

Since $\sum_{i=1}^{n-1} \frac{1}{\sqrt{i}} \leq \int_{0}^{n-1}\frac{1}{\sqrt{x}} dx= 2 \sqrt{n-1}$, we get that

\begin{align*}
OPT(n) &\leq \frac{n(n-1)}{2}+ \delta\cdot \sqrt{2(w-1)} \cdot \left( \sum_{i=1}^{n-1} \sqrt{i} + \sqrt{n-1} \right) \\
& \leq  	\frac{n(n-1)}{2}+ \delta\cdot \sqrt{2(w-1)} \cdot\sum_{i=1}^{n} \sqrt{i} 
\end{align*}

To complete the proof observe that when $w$ grows, $\frac{n(n-1)}{2}$ becomes negligible compared to the second term that grows to infinity with $w$. Thus for any $0<\eps<1$ it holds that for large enough $w$
$$ OPT(n) \leq (1+\eps) \sqrt{2w} \cdot \sum_{i=1}^{n} \sqrt{i}.$$
as we need to prove.
\end{proof}

As a corollary of the proposition we can show that for a fixed $n$ as $w$ approaches infinity the probability of entering the queue is approaching 0: 
\begin{corollary} \label{cor:opt:p-to-0}
Fix $n\geq 2$. For any $\eps >0$ there exists $w_0^n(\eps)$ such that for every $w>w_0^n(\eps)$ it holds that  $\frac{1}{\eps \cdot w}<p_n < \eps$.
\end{corollary}
\begin{proof}
For simplicity we denote $p=p_n$. Recall that by Equation (\ref{eq:optimum}), the recursive formula for the optimal solution (also presented in the beginning of this section), we have that:

\begin{align*}
OPT(n) = \frac{(1-p)^n \cdot n + \sum_{i=1}^{n} p^i (1-p)^{n-i} {n \choose i} \left(w\sum_{j=1}^{i-1} j+ i(n-i) + OPT(n-i)\right)}{1-(1-p)^n}
\end{align*} 

To show that there exists $w_0^n(\eps)$ such that for every $w>w_0^n(\eps)$ it holds that $p_n < \eps$, it suffices to lower bound the second term in the formula by the cost associated with the event that exactly two player enter the queue simultaneously ($i=2$). 
That is:
\begin{align*}
OPT(n) \geq \frac{w\cdot p^2 (1-p)^{n-2} {n \choose 2}}{1-(1-p)^n} \geq \frac{w\cdot p^2 (1-p)^{n-2} {n \choose 2}}{n \cdot p} = w\cdot \frac{p (1-p)^{n-2} (n-1)}{2}
\end{align*}
The second transition is according to the simple auxiliary Lemma \ref{lem:bounds-when-p-small} that can be found in the Appendix.

Observe that the lower bound on $OPT(n)$ is linear in $w$ when $p$ is bounded away from both $0$ and $1$. By Proposition \ref{prop:opt-ub-large-w} we know that the optimal cost is asymptotically sub-linear in $w$. Now, since clearly for large $w$ it cannot be the case that $p>1/2$ we have that $p$ must be approaching $0$, otherwise the inequality will not be satisfied.

Next we show that for any $\eps>0$, for large enough $w$ it holds that $\frac{1}{\eps \cdot w}<p$. This is equivalent to claiming that $w\cdot p$ goes to infinity. This is so as otherwise for some finite constant $K$ it will hold that $K> w\cdot p$ for every $w$, which means that $\frac{1}{p}> \frac{w}{K}$ for every $w$.	
 But since the cost is at least $\frac{1}{p}$ (which is the expected time till one enters the queue, see also Equation (\ref{eq-opt-lowerbound-for}) ), we conclude that the cost must be growing linearly in $w$. But, by the upper bound presented in Proposition \ref{prop:opt-ub-large-w} we know that the optimal cost is asymptotically sub-linear in $w$, a contradiction.
\end{proof}

\subsection{A Lower Bound on $OPT(n,w)$}
\label{sec:OPT-LB}

We prove the following lower bound on $OPT(n,w)$:
\begin{proposition}\label{prop:opt-LB}
Fix any $n\geq 2$ and a positive $0<\eps<1$. There exists $w_0^n(\eps)$ such that for any $w>w_0^n(\eps)$ it holds that   
$$ 
OPT(n,w) \geq (1-\eps) \sqrt{2w} \cdot \sum_{i=1}^{n-1} \sqrt{i}   
$$
\end{proposition}
We outline the proof of the proposition, the complete proof appears in Section \ref{sec-opt-lb}.
We first provide a lower bound to the expected cost of all cases in which at least two agents join the queue, and the optimal cost suffers a cost of at least $w$. The bound implies that for any $\beta \in (0,1)$, 
and any $n$, for large enough $w$ we have:
	$$ 
	OPT(n)
	\geq (1-p)^{n-1}  OPT(n-1) + \frac{1}{p} +  \beta\cdot w\cdot \frac{(n-1)p}{2}
	$$
	
 We then show that for any fixed $\beta>0$ the function $\frac{1}{p} +  \beta\cdot w\cdot \frac{(n-1)p}{2}$ is minimized at $p=\sqrt{\frac{2}{\beta\cdot w\cdot (n-1)}}$ and its value at the minimum is 		$\sqrt{2\cdot \beta\cdot w(n-1)}$.

Now observe that by Corollary \ref{cor:opt:p-to-0} we have that for any $0<\alpha<1$ and any $n\geq 2$ for large enough $w$ it holds that 
$p$ goes to $0$ and hence $(1-p)^{n-1}>\alpha^{\frac{1}{n}}$. We thus pick $\beta$ and $\alpha$ such that $1+\eps = \sqrt{\beta} \cdot \alpha$ and bound $OPT(n)$ as follows:
\begin{align*}
OPT(n)
& \geq \alpha^{\frac{1}{n}} OPT(n-1) + \sqrt{2\cdot \beta\cdot w(n-1)} \geq  \sqrt{2w \cdot \beta} \sum_{m=1}^{n-1} \alpha^{\frac{m-1}{n}} \sqrt{n-m} 
\geq (1+\eps) \sqrt{2w} \cdot \sum_{i=1}^{n-1} \sqrt{i}.   
\end{align*}
	 
and this concludes the proof of the proposition.
}
\full{ 	
	\subsubsection{Proof of the Lower Bound on $OPT(n,w)$} \label{sec-opt-lb}

Recall that we would like to prove that for any fixed $n\geq 2$ and a positive $0<\eps<1$. There exists $w_0^n(\eps)$ such that for any $w>w_0^n(\eps)$ it holds that   
$$ 
OPT(n,w) \geq (1-\eps) \sqrt{2w} \cdot \sum_{i=1}^{n-1} \sqrt{i}   
$$ 
Let $p=p_n$. Define $$f_i(p,n) =  p^i (1-p)^{n-i} {n \choose i} \left(w\sum_{j=1}^{i-1} j+ i(n-i) + OPT(n-i)\right),$$ the total optimal cost when exactly $i$ agents enter the queue at this round.
	
By rearranging the terms of Equation (\ref{eq:optimum}) and using the expressions $f_i(p,n)$ we get
$$ OPT(n) = \frac{n (1-p)^n + p \cdot n \cdot (1-p)^{n-1} \cdot \left((n-1) + OPT(n-1)\right) + \sum_{i=2}^{n} f_i(p,n)}{1-(1-p)^n} $$

	It is not hard to show that for any $n\geq 2$ and any $p\in [0,1]$ it holds that  $1-p\cdot n\leq (1-p)^n \leq 1- p\cdot n + p^2 {n \choose 2}$ (we provide a formal proof in Lemma \ref{lem:bounds-when-p-small} in the appendix). 
	We repeatedly use this to show:	
	\begin{align} \label{eq-opt-lowerbound-for}
	OPT(n) &\geq \frac{n (1-p\cdot n) + p \cdot n \cdot (1-p)^{n-1} \cdot \left((n-1) + OPT(n-1)\right) + \sum_{i=2}^{n} f_i(p,n)}{n\cdot p} \\
	&\geq \frac{1}{p}  -n 
	+  (1-p)^{n-1}  OPT(n-1) + \frac{\sum_{i=2}^{n} f_i(p,n)}{p\cdot n} \\
	&=(1-p)^{n-1}  OPT(n-1) +  \frac{1}{p} + \left(\frac{\sum_{i=2}^{n} f_i(p,n)}{p\cdot n} -n \right)
	\end{align}
	
	In the next two lemmas we lower bound the term $\frac{\sum_{i=2}^{n} f_i(p,n)}{p\cdot n} -n$. First, in Lemma \ref{lem-fn} we provide a lower bound for $\frac{\sum_{i=2}^{n} f_i(p,n)}{p\cdot n}$ and then in Lemma \ref{lem-beta}
	we use this to derive an asymptotic bound on $\frac{\sum_{i=2}^{n} f_i(p,n)}{p\cdot n} -n$ as $w$ approaches infinity.
	\begin{lemma} \label{lem-fn}
		For any $n\geq 2$, 
		it holds that 
		$$\frac{\sum_{i=2}^{n} f_i(p,n)}{p\cdot n} \geq w\cdot \left(\frac{(n-1)p}{2} - {n-1 \choose 2} p^2 \right)$$
	\end{lemma}
	\begin{proof}
		By using Lemma \ref{lem:bounds-when-p-small} from the appendix we  get: 
		
		$$\frac{\sum_{i=2}^{n} f_i(p,n)}{p\cdot n} \geq w\cdot \frac{1-p\cdot n \cdot (1-p)^{n-1}- (1-p)^n}{p\cdot n}\geq $$
		$$ w\cdot \frac{1-p\cdot n \cdot \left(1- p\cdot (n-1) + p^2 {n-1 \choose 2}\right)- \left(1- p\cdot n + p^2 {n \choose 2}\right)}{p\cdot n}
		= w\cdot \left(\frac{(n-1)p}{2} - {n-1 \choose 2} p^2 \right)
		$$ 
	\end{proof}
	
	We conclude: 
	$$ 
	OPT(n)
	\geq (1-p)^{n-1}  OPT(n-1) + \frac{1}{p} +  w\cdot \left(\frac{(n-1)p}{2} - {n-1 \choose 2} p^2 \right) -n 
	$$
	\begin{lemma} \label{lem-beta}
		For any $0<\beta<1$ and any $n$ there exists $w_1^n(\beta)$ such that for every $w>w_1^n(\beta)$ it holds that 
		$$w\cdot \left(\frac{(n-1)p}{2} - {n-1 \choose 2} p^2 \right) -n > \beta \cdot w\cdot \frac{(n-1)p}{2} $$ 
	\end{lemma}
	\begin{proof}	
		The claim is equivalent to saying that for large enough $w$ 
		$$1-\beta>\frac{2n}{w(n-1)p}+p(n-2)$$
		We note that for $n\geq 2$ it holds that $1<\frac{n}{n-1} = 1+\frac{1}{n-1}\leq 2$ 
		thus for the above to hold it is sufficient to show that  for large enough $w$
		$$1-\beta> \frac{4}{w\cdot p}+p(n-2) $$	
		The LHS is a fixed positive constant. The RHS 
		is made out of two terms, both of which  
		tend to $0$ as $w$ grows to infinity, 
		by Corollary \ref{cor:opt:p-to-0}. 

	\end{proof}

	We conclude that for any $\beta<1$ there exists $w_1^n(\beta)$ such that for every $w>w_1^n(\beta)$ we have:
	$$ 
	OPT(n)
	\geq (1-p)^{n-1}  OPT(n-1) + \frac{1}{p} +  \beta\cdot w\cdot \frac{(n-1)p}{2}
	$$
	
	\begin{lemma}
		Fix $\beta>0$. 
		The function $\frac{1}{p} +  \beta\cdot w\cdot \frac{(n-1)p}{2}$ is minimized at $p=\sqrt{\frac{2}{\beta\cdot w\cdot (n-1)}}$ and its value at the minimum is 
		$\sqrt{2\cdot \beta\cdot w(n-1)}$.
	\end{lemma}
	\begin{proof}
		We take the derivative by $p$ and equalize to zero:
		$-\frac{1}{p^2} +  \beta\cdot w\cdot \frac{(n-1)}{2}=0$. 
		For $p\in [0,1]$ the solution is $p=\sqrt{\frac{2}{\beta\cdot w\cdot (n-1)}}$.
		We note that it is clear that this point is the minimum of the function, and that the minimum is not obtained at $0$ or $1$. 
	\end{proof}
	
	Now observe that by Corollary \ref{cor:opt:p-to-0} we have that for any $\alpha<1$ and any $n\geq 2$ there is $w_2^n(\alpha)$ such that for every $w>w_2^n(\alpha)$ it holds that $(1-p)^{n-1}>\alpha^{\frac{1}{n}}$. 
	We thus pick $\beta$ and $\alpha$ such that $1+\eps = \sqrt{\beta} \cdot \alpha$ and bound $OPT(n)$ for every $w>w_0^n(\eps) = \max\{w_1^n(\beta), w_2^n(\alpha)\}$ as follows:
	\begin{align*}
	OPT(n)
	& \geq \alpha^{\frac{1}{n}} OPT(n-1) + \sqrt{2\cdot \beta\cdot w(n-1)} \geq  \sqrt{2w \cdot \beta} \sum_{m=1}^{n-1} \alpha^{\frac{m-1}{n}} \sqrt{n-m} 
	\geq (1-\eps) \sqrt{2w} \cdot \sum_{i=1}^{n-1} \sqrt{i}   
	\end{align*}
	
	This concludes the proof of the proposition. 
}

\bibliographystyle{plain}
\bibliography{n}

\newpage
\appendix

\section{General Strategies} 
\label{app:general-strategies}
In this section we formally define general strategies in our game, strategies that might not be anonymous or stationary. We then show that if all players are playing an anonymous stationary strategy, then if a player has a beneficial deviation to any strategy, he has a deviation to an  anonymous stationary strategy. This means that in the paper we can indeed only consider deviations to 
anonymous stationary strategies when proving that some profile of anonymous stationary strategies is an equilibrium. 

A general strategy for an agent is a mapping from all the histories he might observe to a probability of entering the queue. 
More formally, the history up to time $t-1$ is denoted by $h_{t}$, such an history specifies for every prior time period, the observed action of every player, and the order in which players enter the queue. A strategy $S_i$ for agent $i$ is a function that specifies for every $t$ and every $h_t$, a probability of entry  $q_{i,t}(h_t)\in [0,1]$. Given the profile of strategies $S$, the cost for agent $i$ is the expected cost he suffers when the strategy profile $S$ is played. Agents aim to minimize their expected cost.  

\begin{claim}
	Consider any symmetric profile $S$ of anonymous stationary strategies  (all agents playing the same anonymous stationary strategy).
	If some agent $i$ has a beneficial deviation, then he also has a beneficial deviation to an anonymous stationary strategy.
\end{claim}
\begin{proof}
	Assume that all other players but $i$ are playing symmetric profile $S$ of anonymous stationary strategies. 
	Assume that $S'_i$ is a best response of agent $i$ to $S_{-i}$, and that $S'_i$ is not restricted to be anonymous or stationary.   
	Given such a general strategy $S'_i$ for agent $i$, we show that there exists some anonymous stationary strategy $S''_i$ such that for any history, the cost for agent $i$ with strategies $(S''_i,S_{-i})$ is the same as his cost with strategies $(S'_i,S_{-i})$. 
    This implies that given $S$, if there is a beneficial deviation $S'_i$ for agent $i$, 
    then $i$ also has a beneficial deviation to the anonymous stationary strategy $S''_i$.
	
	We next define an anonymous stationary strategy $S''_i$.
	For any state $(m,k)$ such that $n\geq m\geq 1$ and $n-1\geq k\geq 0$, 
	consider some history $h_t$ that reaches this state $(m,k)$ when the profile is $(S'_i,S_{-i})$, if such an history exists. 
	Define an anonymous stationary strategy $S''_i$ as follows: for any time $t'$ and history $h_{t'}$ for which the reached state is $(m,k)$,	agent $i$ will enter the queue with the same probability as he would if the time is $t$ and the history is $h_t$. We note that this new strategy $S''_i$ only depends on the state, and thus is an anonymous and stationary. 
	
	We next argue that this anonymous stationary strategy $S''_i$ has exactly the same cost as the original strategy $S'_i$, when the others are playing $S_{-i}$. 	
	We first note that the expected cost for agent $i$ that is outside the queue depends only on the state $(m,k)$ reached given the history. That is, if given each of the two histories $h_{t'}$ and $h_{t}$ the state is $(m,k)$, then the cost for agent $i$ when starting from time $t'$, is the same as the cost when starting at time $t$, as $S'_i$ was a best response for $i$ to $S_{-i}$. 
	
	This is so as the cost only depends on current state $(m,k)$ and the Markov chains of future states, which are the same as all other agents are using the same anonymous stationary strategies, and thus an optimal strategy must play in a way that will ensure the same cost at both histories, or otherwise it can be improved by following the strategy in the less costly history, also for the other history. 		
\end{proof}  


\section{Existence of Symmetric Equilibria}
\label{app-eq-existance}
	\begin{theorem}
Fix any $n\geq 2$ and $w>1$. 
The game $G(n,w)$ has a symmetric equilibrium in anonymous stationary strategies. 
\end{theorem}
\begin{proof}	
We will prove by induction that for every sub-game $G(m,k;w)$ such that there are $m\geq 1$ agents outside and $k\geq 0$ agents are in the queue, with $m+k\leq n$, there exists a symmetric equilibrium. 
Observe that this trivially holds for the base case of $(1,k)$ for any $k\geq 0$. Next, we assume that there exists a symmetric equilibrium for any sub-game $G(m,k;w)$ such that $m+k \leq n-1$ and prove that there is a symmetric equilibrium for any sub-game $G(m,k;w)$ such that $m+k = n$. For any $m+k=n-1$ fix some symmetric equilibrium in $G(m,k;w)$. 
Now for any $m',k'$ such that $0\leq m' \leq m$ and $0\leq k' \leq k+m-1-m'$ denote by $c(m',k')$ the cost of an agent when all the agents are playing this equilibrium and there are $m'$ agents waiting to get serviced and $k'$ agent waiting in line.  

To prove that there exists an equilibrium in the sub-game $G(m,k;w)$ with $m+k = n$ we show that there exists probability $q\in [0,1]$ such that in the state $(m,k)$, if all players but $i$ are entering with probability $q$, it is a best response for player $i$ to also enter the queue with probability $q$. This will allow us to find a symmetric equilibrium in the sub-game $G(m,k;w)$ in which in the state $(m,k)$ all agents are each entering with probability $q$, until at least one enters the queue, and then they continue  playing the symmetric equilibrium in each sub-game they get to, which exists by the induction hypothesis. 

Recall that $c^0_{m,k}(q)$ denotes the cost of player $i$ for not entering the queue where the rest of the players are entering with probability $q$ and then are playing some equilibrium strategy, and that $c^1_{m,k}(q)$ denotes the corresponding cost for the case that the player is joining the queue with probability $1$. Observation \ref{obs:eq-costs-1} and \ref{obs:eq-costs-0} present the following bounds on these costs. 
For entering the queue:
\begin{align*} 
c^1_{m,k}(q) = \frac{m-1}{2} \cdot  q \cdot w + k\cdot w. 
\end{align*}
As for not entering the queue, for $k=0$:
\begin{align*}
 c^0_{n,0}(q) &= \frac{1}{1-(1-q)^{n-1}} + \frac{1}{1-(1-q)^{n-1}} \cdot \sum_{i=1}^{n-1} {{n-1}\choose{i}} q^i \cdot (1-q)^{n-1-i} \cdot c(n-i,i-1)
\end{align*}
and for $k\geq 1$
\begin{align*}
c^0_{m,k}(q) &= 1+\sum_{i=0}^{m-1} {{m-1}\choose{i}} q^i \cdot (1-q)^{m-1-i} \cdot c(m-i,k+i-1) 
\end{align*}

Observe that when $q$ tends to $0$, $c^1_{n,0}(q)$ tends to $0$, while $c^0_{n,0}(q)$ tends to infinity. This implies that there exists an $\eps>0$ such that $c^1_{n,0}(\eps)<c^0_{n,0}(\eps)$. Also, note that both $c^1_{m,k}(q)$ and $c^0_{m,k}(q)$ are continuous in $[\eps,1]$. Based on this, we now distinguish between the following three cases:

\begin{itemize}
\item If for every $q\in [\eps,1]$, $c^1_{m,k}(q) < c^0_{m,k}(q)$ then this is also true for $q=1$ and hence there exists a unique exists a symmetric equilibrium in which $q=1$.
\item If for every $q\in [\eps,1]$, $c^1_{m,k}(q) > c^0_{m,k}(q)$. By the definition of $\eps$ this case is only possible for $k\geq 1$. In this case, if for every $q\in [0,\eps]$ it also holds that $c^1_{m,k}(q) > c^0_{m,k}(q)$ then $q=0$ satisfies the symmetric equilibrium requirement. Else, since for $k\geq 1$ both $c^0_{m,k}(q)$ and $c^1_{m,k}(q)$ are continuous in $[0,1]$ by the intermediate value theorem there exists $q^*$ such that $c^1_{m,k}(q^*) = c^0_{m,k}(q^*)$ as required.
\item Otherwise there exists $q^*$ such that $c^1_{m,k}(q^*) = c^0_{m,k}(q^*)$ and this $q^*$ satisfies the symmetric equilibrium requirement. 
\end{itemize}
This completes the proof.
\end{proof} 


\short{\section{Missing proofs from Section \ref{sec:symmetric-NE}} 

 }

\section{Useful Loose Bounds on Entrance Probabilities in any Symmetric Equilibrium}
\label{app-prob-bounds}
In this section we prove Proposition \ref{prp:nash:zero}. We first restate it here:

\begin{proposition} 
	Fix $n\geq 2$. For any $\eps\leq 1$ there exists $w_0^n(\eps)$ such that for any $w\geq w_0^n(\eps)$ in every symmetric equilibrium:
	\begin{itemize}
		\item For any $n'$ such that $ 2 \leq n' \leq n$ it holds that $q_{n',0}=q_{n'} \cdot (n'-1) \leq \eps$.
		\item For any $k\geq 1$ and $m\geq 1 $ such that $k+m \leq n$ it holds that $q_{m,k} = 0$.
	\end{itemize}
\end{proposition}
\begin{proof}
We will prove this by induction on $n$.	
Base case: $n=2$. In this case it is easy to see that $q_{1,1} = 0$ for any $w>1$. Furthermore, by Claim \ref{clm-2players-ne-prob} we have that $q_2 = \sqrt{\frac{2}{w}}$ and hence $(2-1) \cdot q_2 = \sqrt{\frac{2}{w}}$. Thus for every $w>\frac{2}{\eps^2}$ we have that $q_2 \leq \eps$. 
	
	We now assume the claim holds for $n$, and prove that it holds for any $m+k=n+1$. 
	We provide different bounds for $w_0^{n+1}(\eps)$ depending on the value of $\eps$:
	\begin{itemize}
		\item For $\frac{1}{n} <\eps \leq 1$ we will show that $w_0^{n+1}(\eps) = \max\{w_0^{n}(\frac{\eps}{2}),\frac{8n}{\eps}\}$.
		\item For $\eps <\frac{1}{n}$ we will show that $w_0^{n+1}(\eps) = \max\{w_0^{n}(\frac{\eps}{2}),\frac{4}{\eps} \cdot (\frac{2}{\eps \cdot n}+n)\}$.
	\end{itemize}
	
	Observe that in either case we have that $w>w_0^{n}(\frac{\eps}{2})$. This means that by using the induction hypothesis we get that for any $w \geq w_0^{n+1}(\eps)$:
	\begin{itemize}
		\item For any $n'$ such that $ 2 \leq n' \leq n$ it holds that  $q_{n',0}=q_{n'} \cdot (n'-1) \leq \eps/2 \leq \eps$. 
		\item 
		For any $k\geq 1$ and $m\geq 1 $ such that $k+m \leq n$ it holds that $q_{m,k} = 0$.
	\end{itemize}
	
	We are left with showing that for every $w \geq w_0^{n+1}(\eps)$:
	\begin{itemize}
		\item $q_{n+1} \cdot n \leq \eps$.
		\item 
		For any $k\geq 1$ and $m\geq 1 $ such that $k+m \leq n+1$ it holds that $q_{m,k} = 0$.
	\end{itemize}
	
	We show each one in a separate claim:
	\begin{claim}
		For every $w>w_0^{n+1}(\eps)$, $q_{n+1} \cdot n \leq \eps$.
	\end{claim}
	\begin{proof}
		Assume towards contradiction that there exists $w\geq w_1^{n+1}(\eps)$ for which $q_{n+1}(w) \cdot n > \eps$.
		
		Observe that if $q_{n+1}(w) \cdot n > \eps$ then:
		\begin{align*}
		c_{n+1} = \frac{n}{2} \cdot q_{n+1} \cdot w > \frac{\eps}{2} \cdot w
		\end{align*}
		
		On the other hand, we know that:
		\begin{align*}
		c_{n+1} &\leq \frac{1}{1-(1-q_n)^{n}} + \max_{0 \leq i\leq n} c(n-i,i)
		\end{align*}
		By the induction hypothesis we have that $q_{i,n-i} = 0$ for every $i$ such that $1\leq i\leq n$.
		Additionally we know that 	$c_{i}  = \frac{i-1}{2} \cdot q_i \cdot w $. Thus:  
		\begin{align*}
		\max_{0 \leq i\leq n} c(i,n-i) \leq \max_{0 = i\leq n} \left(n-i+c_{i} \right) = \max_{0 \leq i\leq n} \left(n-i+\frac{i-1}{2} \cdot q_i \cdot w \right) 
		\end{align*}
		
		Next, as we assumed that $w\geq w_0^{n}(\frac{\eps}{2})$ by the induction hypothesis we get that for every $i\leq n$ it holds that $(i-1)\cdot q_i \leq \frac{\eps}{2}$. Thus, 
		$\max_{0 \leq i\leq n} c(i,n-i) \leq n+ \frac{\eps}{4} \cdot w$. This in turn implies that:
		\begin{align*}
		c_{n+1} \leq \frac{1}{1-(1-\eps)^{n}} + n+ \frac{\eps}{4} \cdot w.
		\end{align*}
		\comment{
			
			Observe that if $\eps > \frac{1}{n}$ by Lemma \ref{lem-big-prob} for $p_0 = \frac{1}{n}$ we have that $\frac{1}{1-(1-\eps)^{n}} \leq \frac{e}{e-1}$. 
			Furthermore, by our choice of $w_1^{n+1}(\eps)$ we have that $w \geq \frac{8n}{\eps}$ and hence:
			\begin{align*}
			\frac{\eps}{4} \cdot w \geq \frac{\eps}{4} \cdot \frac{8n}{\eps} = 2n
			\end{align*}
		}
		
		To reach a contradiction we should show that: 
		\begin{align*} 
		\frac{\eps}{2} \cdot w \geq \frac{1}{1-(1-\eps)^{n}} + n+ \frac{\eps}{4} \cdot w 
		\end{align*}
		The above inequality holds if and only if:
		\begin{align} \label{eq-contradiction}
		\frac{\eps}{4} \cdot w \geq \frac{1}{1-(1-\eps)^{n}} + n.
		\end{align}
		
		Observe that if $\eps > \frac{1}{n}$ by Lemma \ref{lem-big-prob} for $p_0 = \frac{1}{n}$ we have that $\frac{1}{1-(1-\eps)^{n}} \leq \frac{e}{e-1}$. Furthermore, by our choice of $w_1^{n+1}(\eps)$ we have that $w \geq \frac{8n}{\eps}$ and hence:
		\begin{align*}
		\frac{\eps}{4} \cdot w \geq \frac{\eps}{4} \cdot \frac{8n}{\eps} = 2n
		\end{align*}
		Since $2n > \frac{e}{e-1} + n$ we get that Equation (\ref{eq-contradiction}) holds and a contradiction is reached.
		
		For the case that $\eps< \frac{1}{n}$ by Lemma \ref{lem-small-probability} we have that $\frac{1}{1-(1-\eps)^{n}} \leq \frac{2}{\eps n}$. By the assumption that $w \geq \frac{4}{\eps} \cdot (\frac{2}{\eps \cdot n}+n)$, we have that:
		\begin{align*}
		\frac{\eps}{4} \cdot w > \frac{\eps}{4} \cdot \frac{4}{\eps} \cdot \left(\frac{2}{\eps \cdot n}+n\right) = \frac{2}{\eps \cdot n}+n\geq \frac{1}{1-(1-\eps)^{n}} +n
		\end{align*}
		and by Equation (\ref{eq-contradiction}) a contradiction is reached.
	\end{proof}
	
	\begin{claim}
		For every $w>w_0^{n+1}(\eps)$ and for any $k\geq 1, m\geq 1$ s.t. $k+m = n+1$ it holds that $q_{m,k} = 0$.
	\end{claim}
	\begin{proof}
		Assume towards contradiction that there exists $w \geq w_0^{n+1}(\eps)$ for which $q_{m,k} > 0$. This in particular implies that the player is indifferent between entering for sure and not entering. Hence, his cost equals to the cost of entering which is:
		\begin{align} \label{eq-cmk}
		c(m,k) = \frac{m-1}{2} \cdot q_{m,k} \cdot w + k\cdot w \geq k\cdot w > w
		\end{align} 
		where the last inequality is due to the fact that $k\geq 1$. On the other hand, his cost also equals the cost for not entering the queue which is:
		\begin{align*}
		c(m,k) \leq 1 + \max_{0\leq i \leq m}  c(m-i,k-1+i) &=1+\max_{0\leq i \leq m}  (c_i + m-i+k-1) \\
		&\leq m+k + \max_{0\leq i \leq m} 
		\left(\frac{i-1}{2} \cdot q_{i}\cdot w \right)
		\end{align*}
		where the equality is by the induction hypothesis stating that $q_{m-i,k+1-i} = 0$ for every $0 \leq i \leq m$.

		Note that no matter what is the value of $\eps$ (as long as it is less than $1$), we have that $w>w_0^{n+1}(\eps) > \max\{ w_0^{n}(\frac{\eps}{2}), 4n \}$. This implies that we can use the induction hypothesis and get that for any $i\leq n$, $q_i(i-1) \leq \eps/2$. 
		
		Next by our assumption that $w \geq 4n > 2(n+1)$ we get that $m+k \leq n+1 \leq \frac{w}{2}$. Putting this together we get that $c(m,k) \leq \frac{1}{2}w + \frac{\eps^2}{2} \cdot w \leq w$. This is in contradiction to Equation (\ref{eq-cmk}) showing that $c(m,k) > w$.
	\end{proof}
	
	This concludes the proof of the proposition. 
\end{proof}
 
\short{ \section{Missing proofs from Section 5}   \label{app:symmetric-OPT}

 }
\short{\subsection{Proof of Corollary  \ref{cor:ratio-opt-sw}:}}
\full{\section{Proof of Corollary  \ref{cor:ratio-opt-sw}:}}  \label{app-cor-opt}
We next restate and prove Corollary \ref{cor:ratio-opt-sw}:
\\
	Fix any $\delta_1>0$. 
	There exists $n_0(\delta_1)$ such that for any $n>n_0(\delta_1)$ there exist $w^{n}(\delta_1)$ such that for any $w>w^{n}(\delta_1)$ it holds that for any symmetric equilibrium $S$   
	$$ 
	(1-\delta_1) \frac{3}{2\sqrt{2}}  \leq \frac{C_{n,w}(S)}{OPT(n,w)} \leq (1+\delta_1) \frac{3}{2\sqrt{2}} 
	$$
	We note that $\frac{3}{2\sqrt{2}}\approx 1.061 $.		 	 
\\
\begin{proof} 
	First we observe that by the Euler–-Maclaurin formula: $\frac{2}{3} k\sqrt{k} < \sum_{i=1}^{k} \sqrt{i} < \frac{2}{3} k\sqrt{k}+\sqrt{k}$. This together with Theorem \ref{thm:opt} implies that, for any $\eps>0$, there exists $w_1^n(\eps)$ such that for every $w>w_1^n(\eps)$:
	\begin{align*}
	(1-\eps) \sqrt{2w} \cdot \frac{2}{3} (n-1)\sqrt{n-1} \leq OPT(n,w) \leq (1+\eps) \sqrt{2w} \left(\frac{2}{3} \cdot n\sqrt{n}+\sqrt{n}\right)
	\end{align*}
	
	By Theorem \ref{thm:opt-large-w}	
	for any $\eps>0$, there exists $w_0^n(\eps)$ such that for any $w \geq w_0^n(\eps)$ in every symmetric equilibrium: 	
	\begin{align*}
	(1-2\eps)^{n-1} \cdot \sqrt{w\cdot n- w\cdot 2\ln n}\cdot n \leq C_{n,w}(S) \leq 
	(1+2\eps) \cdot \sqrt{w \cdot n + 2w\sqrt{n-1}} \cdot n
	\end{align*}

	Let $w_{max}^n = \max \{w_0^n(\eps),w_1^n(\eps) \}$. Then, for every $w>w_{max}^n$: 
	
	\begin{align*}
	\frac{(1-2\eps)^{n-1} \cdot \sqrt{w\cdot n- w\cdot 2\ln n}\cdot n }{(1+\eps) \sqrt{2w} \left(\frac{2}{3} \cdot n\sqrt{n}+\sqrt{n}\right) }
	\leq \frac{C_{n,w}(S)}{OPT(n,w)}
	\leq  \frac{(1+2\eps) \cdot \sqrt{w \cdot n + 2w\sqrt{n-1}} \cdot n}{(1-\eps) \sqrt{2w} \cdot \frac{2}{3} (n-1)\sqrt{n-1}}
	\end{align*}
	
	\begin{align*}
	\frac{(1-2\eps)^{n-1}}{(1+\eps)} \cdot \frac{3}{2\sqrt{2}} \cdot \frac{\sqrt{n- 2\ln n}}{ \sqrt{n}+\frac{3}{2\sqrt{n}} }
	\leq \frac{C_{n,w}(S)}{OPT(n,w)}
	\leq  \frac{(1+2\eps)}{(1-\eps)} \cdot \frac{3}{2\sqrt{2}}  \cdot \frac{ \sqrt{n + 2\sqrt{n-1}} \cdot n}{ (n-1)\sqrt{n-1}}
	\end{align*}

	Finally, note that $\lim_{n \rightarrow \infty } \frac{\sqrt{n- 2\ln n}}{ \sqrt{n}+\frac{3}{2\sqrt{n}} }  = 1$ and $\lim_{n \rightarrow \infty } \frac{ \sqrt{n + 2\sqrt{n-1}} \cdot n}{ (n-1)\sqrt{n-1}} = 1$. This implies that for every $0<\delta<1$, there exists $n_0(\delta)$ such that for any $n>n_0(\delta)$:  $\frac{\sqrt{n- 2\ln n}}{ \sqrt{n}+\frac{3}{2\sqrt{n}} } \geq 1-\delta$ and $\frac{ \sqrt{n + 2\sqrt{n-1}} \cdot n}{ (n-1)\sqrt{n-1}} \leq 1+ \delta$. The proof is completed by picking $\eps(\delta)$ such that $\frac{(1-2\eps(\delta))^{n-1}}{(1+\eps(\delta))} \geq 1-\delta$ and $\frac{(1+2\eps(\delta))}{(1-\eps(\delta))} \leq 1+ \delta$. This implies that the proof holds for any $w>w^{n}(\delta) =w_{max}^n(\eps(\delta)) $. 
	The proof follows from picking $\delta>0$ to satisfy $(1-\delta)^2 = 1-\delta_1$.
\end{proof}	

\section{Auxiliary Claims}

\begin{lemma}\label{lem:bounds-when-p-small}
For any $n\geq 2$ and any $p\in [0,1]$ it holds that  $1-p\cdot n\leq (1-p)^n \leq 1- p\cdot n + p^2 {n \choose 2}$. 
\end{lemma}
\begin{proof}
	We prove that $1-p\cdot n\leq (1-p)^n$ by induction on $n$. Clearly the claim is true for $n=1$. Assume it holds for $n$, we show that it holds for $n+1$. Indeed by the induction hypothesis 
	$$(1-p)^{n+1}= (1-p)\cdot (1-p)^n\geq (1-p)(1-n\cdot p) = 1-(n+1)\cdot p+n\cdot p^2\geq 1-(n+1)\cdot p$$

	We next show that $(1-p)^n\leq 1- p\cdot n + p^2 {n \choose 2}$.
	We again prove the claim by induction on $n$. Clearly the claim is true for $n=2$. 
	Assume it holds for $n$, we show that it holds for $n+1$. 
	Indeed by the induction hypothesis 
	$$(1-p)^{n+1}= (1-p)\cdot (1-p)^n\leq (1-p)\left( 1- p\cdot n + p^2 {n \choose 2}\right)= 
	1- p\cdot (n+1) + p^2 {n+1 \choose 2}- p^3 {n\choose 2}
	$$
	which implies the claim since $p\geq 0$.  
\end{proof}

\begin{lemma} \label{lem-small-probability}
If $p<\frac{2}{n-1}$ then $\frac{1}{1-(1-p)^n} \leq \frac{1}{2-p \cdot (n-1)} \cdot \frac{2}{p\cdot n}$.
\end{lemma}
\begin{proof}
Observe that by Lemma \ref{lem:bounds-when-p-small} we have that $(1-p)^n \leq 1- p\cdot n + p^2 \cdot {n \choose 2}$. This implies that:
\begin{align*}
1-(1-p)^n &\geq p\cdot n - p^2 \cdot {n \choose 2} \\
1-(1-p)^n &\geq p\cdot n \cdot (1-\frac{p(n-1)}{2})
\end{align*}
Because we assumed that $p<\frac{2}{n-1}$, we have that $1-\frac{p(n-1)}{2} > 0$, hence we can divide and get that:
\begin{align*}
\frac{1}{1-(1-p)^n} \leq \frac{1}{1-\frac{p \cdot (n-1)}{2}} \cdot \frac{1}{p\cdot n}  \implies \frac{1}{1-(1-p)^n} \leq \frac{1}{2-p \cdot (n-1)} \cdot \frac{2}{p\cdot n}
\end{align*}
\end{proof}

\begin{lemma}  \label{lem-big-prob}
If $p\geq p_0$, then, $\frac{1}{1-(1-p)^{n}} \leq \frac{e^{n \cdot p_0}}{e^{n \cdot p_0} -1}$.
\end{lemma}
\begin{proof}
Observe that:
$(1-p)^{n}\leq (1-p_0)^{n} =\left( (1-p_0)^{\frac{1}{p_0}} \right)^{n \cdot p_0} \leq e^{-n \cdot p_0}$. This implies that $\frac{1}{1-(1-p)^{n}} \leq \frac{1}{1-e^{-n \cdot p_0}} = \frac{e^{n \cdot p_0}}{e^{n \cdot p_0} -1}$.
\end{proof}

\begin{lemma} \label{lem-sqrt}
The following holds for any $n\geq 1$:
\begin{align*}
\sqrt{n}+1 +\frac{1}{2(\sqrt{n}+1)} \leq 1 + \sqrt{n+1}
\end{align*}
\end{lemma}
\begin{proof}
Observe that $\frac{1}{2(\sqrt{n}+1)} < \frac{1}{2(\sqrt{n+1})}$, hence it suffices to show that:
\begin{align*}
\sqrt{n} +\frac{1}{2\sqrt{n+1}} \leq \sqrt{n+1}.
\end{align*}
This is a known inequality and it is easy to verify its correctness.
\end{proof}

\begin{lemma}
	\label{lem:sum-sqrt-rec}
	For any $m\geq 1$ it holds that $\sum_{i=1}^{m} \frac{1} {1+\sqrt{i}} \geq 2\left(\sqrt{m+1}- \ln (1+\sqrt{m+1}) -1 + \ln 2 \right)	$ 
\end{lemma} 
\begin{proof}
	It holds that 
\begin{align*}
	\sum_{i=1}^{m} \frac{1} {\sqrt{i}+1} & \geq \int_{1}^{m+1} \frac{1} {\sqrt{x}+1}  dx  =  2(\sqrt{x}-\ln (\sqrt{x}+1)) 
	\biggr\rvert_1^{m+1} 
	\\ & \geq  2(\sqrt{m+1}- \ln (\sqrt{m+1}+1)-1 + \ln 2 )
\end{align*}
\end{proof}

\end{document}